\documentclass[11pt]{article}
\usepackage[letterpaper, margin=1in]{geometry}
\usepackage[utf8]{inputenc}
\def\PrintMode{0}

\usepackage{tikz}
\usetikzlibrary{positioning}
\usepackage{pgffor}
\usetikzlibrary{decorations.pathreplacing}
\usepackage{expl3}
\usepackage[T1]{fontenc}
\usepackage{graphicx}
\usepackage{float}
\usepackage[lined,ruled,linesnumbered]{algorithm2e}
\usepackage{xspace}
\usepackage[utf8]{inputenc}
\usepackage[noadjust]{cite}
\usepackage{amsmath, amssymb, amsfonts, amsthm}
\usepackage{enumerate}
\usepackage{xcolor}
\usepackage[all]{xy}
\usepackage[linktocpage=true,pagebackref=true,colorlinks,linkcolor=magenta,citecolor=blue,bookmarks,bookmarksopen,bookmarksnumbered]{hyperref}
\usepackage{bm}
\usepackage{bbm}
\usepackage[normalem]{ulem}
\usepackage{paralist}
\usepackage{gensymb}
\usepackage{thmtools}
\usepackage{rotating}
\usepackage{mdframed}
\usepackage{multirow}
\usepackage{appendix}
\usepackage{tabularx}
\usepackage{verbatim}
\usepackage{mathrsfs}
\usepackage{indentfirst}
\usepackage{mleftright}
\usepackage[nottoc]{tocbibind}
\usepackage{complexity}
\usepackage{tablefootnote}
\usepackage{epigraph}
\usepackage{autobreak}
\usepackage{todonotes}

\ifnum\PrintMode=1

\usepackage{savetrees}
\else
\fi

\graphicspath{{./figs/}}

\usepackage[capitalize]{cleveref}
\newtheorem{theorem}{Theorem}[section]
\newtheorem{lemma}[theorem]{Lemma}
\newtheorem{proposition}[theorem]{Proposition}
\newtheorem{corollary}[theorem]{Corollary}
\newtheorem{claim}[theorem]{Claim}

\newtheorem{fact}[theorem]{Fact}
\newtheorem{question}[theorem]{Question}

\newtheorem{conjecture}[theorem]{Conjecture}

\theoremstyle{definition}
\newtheorem{definition}[theorem]{Definition}

\newtheorem{remark}[theorem]{Remark}
\newtheorem*{remark*}{Remark}

\renewcommand*\backref[1]{\ifx#1\relax \else (cit.~on p.~#1) \fi} 

\makeatletter
\def\moverlay{\mathpalette\mov@rlay}
\def\mov@rlay#1#2{\leavevmode\vtop{%
		\baselineskip\z@skip \lineskiplimit-\maxdimen
		\ialign{\hfil$\m@th#1##$\hfil\cr#2\crcr}}}
\newcommand{\charfusion}[3][\mathord]{
	#1{\ifx#1\mathop\vphantom{#2}\fi
		\mathpalette\mov@rlay{#2\cr#3}
	}
	\ifx#1\mathop\expandafter\displaylimits\fi}
\makeatother

\renewcommand{\poly}{\mathrm{poly}}

\renewcommand{\emptyset}{\varnothing}

\newcommand{\cM}{\mathcal{M}}
\newcommand{\rank}{\ensuremath{\operatorname{rank}}}
\newcommand{\cV}{\mathcal{V}}
\newcommand{\cI}{\mathcal{I}}

\newcommand{\MR}{\mathsf{MR}}
\newcommand{\BR}{\mathsf{BR}}

\newlang{\MCSP}{MCSP}
\newlang{\MFSP}{MFSP}
\newlang{\MKtP}{MKtP}
\newlang{\MKTP}{MKTP}
\newlang{\itrMCSP}{itrMCSP}
\newlang{\itrMKTP}{itrMKTP}
\newlang{\itrMINKT}{itrMINKT}
\newlang{\MINKT}{MINKT}
\newlang{\MINK}{MINK}
\newlang{\MINcKT}{MINcKT}
\newlang{\CMD}{CMD}
\newlang{\DCMD}{DCMD}
\newlang{\CGL}{CGL}
\newlang{\PARITY}{PARITY}
\newlang{\Empty}{\textsc{Empty}}
\newlang{\Avoid}{\textsc{Avoid}}
\newlang{\Sparsification}{\textsc{Sparsification}}
\newlang{\HamEst}{\mathsf{HammingEst}}
\newlang{\HamHit}{\mathsf{HammingHit}}
\newlang{\CktEval}{\textsc{Circuit-Eval}}
\newlang{\Hard}{\textsc{Hard}}
\newlang{\cHard}{\textsc{cHard}}
\newlang{\CAPP}{CAPP}
\newlang{\GapUNSAT}{GapUNSAT}
\newlang{\OV}{OV}
\renewlang{\PCP}{PCP}
\newlang{\PCPP}{PCPP}
\newclass{\Avg}{Avg}
\newclass{\ZPEXP}{ZPEXP}
\newclass{\DLOGTIME}{DLOGTIME}
\newclass{\ALOGTIME}{ALOGTIME}
\newclass{\ATIME}{ATIME}
\newclass{\SZKA}{SZKA}
\newclass{\Laconic}{Laconic\text{-}}
\newclass{\APEPP}{APEPP}
\newclass{\SAPEPP}{SAPEPP}

\newclass{\TFSigma}{TF\Sigma}
\newclass{\NTIMEGUESS}{NTIMEGUESS}

\newlang{\Formula}{Formula}
\newlang{\THR}{THR}
\newlang{\EMAJ}{EMAJ}
\newlang{\MAJ}{MAJ}
\newlang{\SYM}{SYM}
\newlang{\DOR}{DOR}
\newlang{\ETHR}{ETHR}
\newlang{\Midbit}{Midbit}
\newlang{\LCS}{LCS}
\newlang{\TAUT}{TAUT}
\newlang{\Poly}{\text{-}Poly}

\newcommand{\one}{\textbf{1}}
\newcommand{\Gr}{\operatorname{Gr}}
\newcommand{\Sch}{\operatorname{Sch}}

\newcommand{\N}{\mathbb{N}}

\newcommand{\F}{\mathbb{F}}

\renewcommand{\R}{\mathbb{R}}

\newcommand{\row}{\mathsf{row}}
\newcommand{\col}{\mathsf{col}}

\DeclareMathOperator*{\Span}{\mathrm{Span}}

\newcommand{\Ldist}{\mathcal{L}_{\mathsf{Dist}}}

\newcommand{\cN}{\mathcal{N}}

\usepackage{adjustbox}

\newcommand{\RS}{\mathsf{RS}}

\renewcommand{\epsilon}{\varepsilon}
\newcommand{\eps}{\epsilon}

\usepackage{diagbox}

\newcommand{\cC}{\mathcal C}

\newcommand{\Q}{\mathbb Q}
\renewcommand{\R}{\mathbb R}
\renewcommand{\C}{\mathbb C}

\usepackage[draft, inline,marginclue,index]{fixme}  
\fxsetup{theme=color,mode=multiuser}

\definecolor{color1}{RGB}{46,134,193}

\definecolor{color7}{RGB}{128,0,128}
\FXRegisterAuthor{yeyuan}{ayeyuan}{\color{red}Yeyuan}

\definecolor{color3}{RGB}{255,128,0}
\FXRegisterAuthor{zihan}{azihan}{\color{blue}Zihan}

\definecolor{color5}{RGB}{128,128,128}
\FXRegisterAuthor{task}{atask}{\colorbox{color5}{\color{white}Task}}

\newif\ifmynotes
\mynotestrue

\title{From Random to Explicit via Subspace Designs With Applications to Local Properties and Matroids}
\date{}
\author{Joshua Brakensiek\thanks{University of California, Berkeley. \href{mailto:josh.brakensiek@berkeley.edu}{\texttt{josh.brakensiek@berkeley.edu}}} 
\and
Yeyuan Chen\thanks{Department of EECS, University of Michigan, Ann Arbor. \href{mailto:yeyuanch@umich.edu}{\texttt{yeyuanch@umich.edu}}}  
\and
Manik Dhar\thanks{Department of Mathematics, Massachusetts Institute of Technology. \href{mailto:dmanik@mit.edu}{\texttt{dmanik@mit.edu}}} 
\and
Zihan Zhang\thanks{Department of Computer Science and Engineering, The Ohio State University. \href{mailto:zhang.13691@buckeyemail.osu.edu}{\texttt{zhang.13691@osu.edu}}}  
}

\begin{document}
\maketitle

\begin{abstract}
In coding theory, a common question is to understand the threshold rates of various local properties of codes, such as their list decodability and list recoverability. A recent work Levi, Mosheiff, and Shagrithaya (FOCS 2025) gave a novel unified framework for calculating the threshold rates of local properties for random linear and random Reed--Solomon codes.

In this paper, we extend their framework to studying the local properties of subspace designable codes, including explicit folded Reed-Solomon and univariate multiplicity codes. Our first main result is a local equivalence between random linear codes and (nearly) optimal subspace design codes up to an arbitrarily small rate decrease. We show any local property of random linear codes applies to all subspace design codes. As such, we give the first explicit construction of folded linear codes that simultaneously attain all local properties of random linear codes. Conversely, we show that any local property which applies to all subspace design codes also applies to random linear codes.

Our second main result is an application to matroid theory. We show that the correctable erasure patterns in a maximally recoverable tensor code can be identified in deterministic polynomial time, assuming a positive answer to a matroid-theoretic question due to Mason (1981). This improves on a result of Jackson and Tanigawa (JCTB 2024) who gave a complexity characterization of $\mathsf{RP} \cap \mathsf{coNP}$ assuming a stronger conjecture. Our result also applies to the generic bipartite rigidity and matrix completion matroids.

As a result of additional interest, we study the existence and limitations of subspace designs. In particular, we tighten the analysis of family of subspace designs constructioned by Guruswami and Kopparty (Combinatorica 2016) and show that better subspace designs do not exist over algebraically closed fields.
\end{abstract}

\pagebreak

\section{Introduction}
One of the biggest success in modern mathematics and theoretical computer science is use randomness to understand discrete and continuous structures. In combinatorics, the probabilistic method~\cite{prob}, pioneered by Erd\H{o}s, has been used to prove existence of various mathematical objects, such as Ramsey graphs, lossless expanders, and many more. In algorithm design and computational complexity theory, the use of randomness has been the key to providing  efficient solutions to many computational problems, such as polynomial identity testing (PIT).

However, for both theoretical and practical considerations, it is often desirable to remove the use of randomness through \emph{derandomization}. From this perspective, randomization can be viewed as the first step toward the resolution of a mathematical problem. For example, in combinatorics, we often seek \emph{explicit constructions} of mathematical objects we only know to exist via the probablistic method. Likewise, in algorithm design, often one seeks a \emph{deterministic} algorithm after a randomized algorithm has been identified.

However, in many cases, derandomization proves to be quite formidable, and the existence of a derandomization is often still remain open. To name two examples from combinatorics and algorithm design, respectively, first it is an easy exercise to check that random graphs are lossless two-sided vertex expanders, but explicit constructions were long-standing open problems until two breakthroughts \cite{zigzag} (one-sided expansion) and \cite{two-sided} (two-sided expansion). Second, randomized PIT algorithm easily follows from Schwarz--Zippel lemma, but a deterministic PIT algorithm remains one of the most significant open problems in theoretical computer science, and perhaps the most key stepping stone toward proving that $\mathsf{P} = \mathsf{BPP}$.

In this paper, we use \emph{subspace designs} introduced by Guruswami and Xing \cite{guruswami2013list} in novel ways to achieve derandomizations in two different problems from combinatorial constructions and algorithm design respectively: \emph{Local properties of linear codes} and \emph{Independence testing of birigidity and maximally recoverability matroids}. We now give a brief overview of each of the subfields of mathematics and computer science our work touches upon. 

\paragraph{Subspace Designs.} The concept of \emph{subspace designs} was first introduced by Guruswami and Xing~\cite{guruswami2013list} to (randomly) construct the first family of positive-rate rank-metric codes that are list-decodable beyond half their minimum distance, as well as efficiently list-decodable Hamming-metric codes of optimal rate over constant-size alphabets with nearly constant list size. Informally, a subspace design is a family of linear subspaces with the property that any subspace of fixed dimension has only small total overlap with the family (see \cref{def:subspace-design} for more details).

While the probabilistic method establishes the existence of subspace designs with both large size and dimension~\cite{guruswami2013list,guruswami2016explicit}, Guruswami and Kopparty~\cite{guruswami2016explicit} subsequently provided the first explicit construction attaining the similar parameters over large fields via the idea of folded Reed--Solomon codes and univariate multiplicity codes. As an application, these explicit subspace designs yield a full derandomization of the code constructions from~\cite{guruswami2013list}. Later, by employing techniques from algebraic function fields, Guruswami, Xing, and Yuan \cite{guruswami2018subspace} constructed explicit subspace designs with comparable parameters over significantly smaller fields.
Besides their use in list-decodable code constructions, subspace designs have also been instrumental in constructing explicit lossless dimension expanders of constant degree~\cite{guruswami2021lossless}, which can be viewed as the linear-algebraic analogue of classical expander graphs. For a more fine-grained historical overview, we refer the reader to the recent survey by Santonastaso and Zullo~\cite{santonastaso2023subspace}.

All the aforementioned list-decodable codes (e.g.,~\cite{guruswami2013list,guruswami2016explicit,guruswami2022optimal}) that make use of the notion of subspace designs are subcodes of classical algebraic codes such as Reed–Solomon, folded Reed–Solomon, algebraic–geometry, univariate multiplicity, and Gabidulin codes. However, a major recent turning point was achieved by the work of Chen and Zhang \cite{cz24}, who proved that explicit folded Reed–Solomon and univariate multiplicity codes themselves attain optimal list-decodability. In particular, they showed that any code admitting sufficiently good subspace designs enjoys near-optimal list-decodability (see \cite[Appendix~B]{cz24} for further details). These codes, known as \emph{subspace designable codes} in \cite{cz24}, will play a central role in our paper.

\paragraph{Local Properties of Linear Codes.}
List-decodability and list-recoverability are two properties people care the most in the theory of error-correcting codes. Given a code $\mathcal{C}\subseteq \Sigma^n$ over alphabet $\Sigma$, a radius $\rho$, an input list size $\ell$, and an output list size $L$, we say $\mathcal{C}$ is $(\rho,\ell,L)$ list-recoverable iff for any received table $S_1,\dots,S_n\in\binom{\Sigma}{\leq \ell}$, there are at most $L$ codewords $c\in\mathcal{C}$ such that at most $\rho$ fraction of coordinates $i\in[n]$ satisfy $c_i\notin S_i$. List-decoding is a special case of list-recovery when $\ell=1$. Motivated by understanding these two basic combinatorial properties of various families of error-correcting codes, Mosheiff, Resch, Ron-Zewi,
Silas, and Wootters \cite{mrr20} introduced a more generalized framework \emph{local property}. In this broader perspective, One can define the notion \emph{local profile} as the forbidden pattern of codes that can be certified by a small list of codewords. A \emph{local property} $\mathcal{P}$ is defined as a family of local profiles, and a code has the property $\mathcal{P}$ iff it does not contain any certificate for any local profile in $\mathcal{P}$. In this language, both list-decodability and list-recoverability are just special cases of local properties.

The study of local properties mostly starts from analyzing random linear codes, which can be seen the ``(nearly) optimal benchmark''\footnote{In some cases non-linear codes can provably outperform linear codes. However, linearity of codes is useful in many applications, and many explicit constructions of specific families of codes are linear. Another work \cite{rclcl} builds the framework of local properties for random (non-linear) codes.} in many scenarios. For any local property, people first characterize its \emph{rate threshold} below which random linear codes satisfy the local property with high probability. Then, people make use of this characterization to show that various specific families of linear codes has  (almost) the same rate threshold. In particular, \cite{mrr20} showed that random LDPC codes behaves the same as random linear codes. Guruswami, Li, Mosheiff, Resch, Silas, and Wootters \cite{guru21} use this framework to show interesting lower bounds for list-decoding and list-recovery. After that, Guruswami and Mosheiff \cite{gm22} proved randomly punctured low-biased codes locally simulate random linear codes. Most recently, a great work by Levi, Mosheiff and Shagrithaya \cite{lms25} generalized the previously mentioned characterization of rate thresholds to \emph{large alphabet}, and used this novel characterization to show random Reed--Solomon codes are local equivalence to random linear codes.  

However, all of the above results are about \emph{randomness-efficient} families of linear codes. Although compared to uniformly random linear codes, these certain families of linear codes requires fewer random bits to sample from, they are not \emph{explicit constructions}. For some special cases of local properties such as list-decoding, there are explicit constructions with the same quality as random linear codes \cite{ST23,bdg24b,cz24,jmst25}. However, explicit constructions that match random linear codes in terms of general local properties, or even the special case of list-recoverability, are unknown before. In this paper we provide the first explicit construction that simulate all local properties of random linear codes and
 address this open problem.

\paragraph{Maximally Recoverability and Rigidity Matroids.}

The quest for explicit constructions matching random ones has deep connections to \emph{matroid theory}. Informally, matroids are combinatorial objects which generalize the behavior of collections of vectors in a vector space (see \Cref{sec:matroid} for a formal definition). We now describe two families of matroids which, while seemingly unrelated, both demand an improved theory of derandomization: maximal recoverability and structural rigidity.

In the distributed storage community, much effort over the last couple of decades \cite{chen2007maximally,huang2013pyramid,BlaumPS13,sathiamoorthy2013xoring,gopalan2014explicit,papailiopoulos2014locally,hu2016new,Gopalan2016,kane2019independence,martinez2019universal,gopi2020maximally,holzbaur2021correctable,cai2021construction,kong2021new,gopi2022improved,bgm2021mds,dhar2023construction,brakensiek2025improved} has been put into studying data storage solutions which information-theoretically maximize their resilience from data loss. Error-correcting codes which satisfy such problems are described as \emph{maximum recoverable (MR)}. In a strong sense, the construction of MR codes is a derandomization problem, as a random code can correct any particular error with high probability, but guaranteeing tolerance to all possible errors simultaneously requires an explicit construction.

From a mathematical perspective, a rather fascinating family of MR codes is the MR \emph{tensor} codes introduced by Gopalan et al.~\cite{Gopalan2016}, where the error-correcting code is the tensor product of two linear codes. Unlike other architectures for which MR has been studied, we even lack a basic understanding of which failures (i.e., erasure errors) we should be able to recover from (e.g., \cite{holzbaur2021correctable}), much less explicit constructions of MR tensor codes~\cite{bdg24b}. Recently, the construction of maximally recoverable tensor codes has also been linked to optimal list-decodable codes~\cite{ST23,brakensiek2023generic,Alrabiah2025Random}. A comprehensive understanding of derandomizing the properties of random linear codes could be of much use in answering such questions.

In combinatorics, a long-standing area of study is \emph{structural rigidity theory}.\footnote{In this paper ``rigidity'' always refers to \emph{geometric} structural rigidity, there is no known link between these questions and the well-studied \emph{matrix rigidity} problems in TCS (e.g., \cite{ramya2020recent}).} A landmark question open question in rigidity theory is that of \emph{graph rigidity}. Consider an undirected graph $G = (V,E)$ and an ambient space $\R^d$, when does there an embedding\footnote{We assume that $\psi$ maps the vertices to ``general position,'' no three points on a line, etc.} $\psi : V \to \R^d$ such that the graph cannot be flexed (i.e., an infinitesimal isometry which is not a translation/rotation)? If such an embedding exists, a random embedding (e.g., random sampled from a multivariate normal distribution) will work probability 1. A simple description for the case $\R^2$ has been known for nearly 100 years~\cite{PollaczekGeiringer,Laman}. However, finding deterministic polynomial time description for rigidity in $\R^3$ has been open since the days of James Clerk Maxwell~\cite{Maxwell, graver1993combinatorial}!

Although MR and structural rigidity are seemingly unrelated, Brakensiek, Dhar, Gao, Gopi, and Larson~\cite{brakensiek2024Rigidity} proved a number of equivalence results between maximal recoverability and structural rigidity. For example, understanding the maximal recoverability of a linear code tensored with itself is (roughly speaking) to this classical graph rigidity question. For our purposes, we build on a different equivalence proven by \cite{brakensiek2024Rigidity} that the maximal recoverability of the tensor product of two different linear codes is equivalent to understanding the \emph{bipartite rigidity} (birigidity) of graphs, an adaptation of classical rigidity proposed by Kalai, Nevo, and Novik~\cite{kalai2015Bipartite}.

We seek to be the first work which takes advantage of this new-found equivalence by using our ``random-to-explicit'' perspective to make simultaneous progress in the theory of both maximal recoverability and rigidity. We point the interested reader to \Cref{subsec:matroid-results} and \Cref{sec:matroid} for a more details. We also recommend the recent survey by Cruickshank, Jackson, Jord{\'a}n, and Tanigawa for many rich connections between rigidity matroids and other problems~\cite{cruickshank2025rigidity}.

\subsection{Our Results}
In order to formulate our main results, we first introduce the notions of subspace designs \cite{guruswami2013list,guruswami2016explicit} and subspace designable codes \cite{guruswami2016explicit,cz24}.

\begin{definition}[{Subspace Design, \cite[Definition 3]{guruswami2016explicit}}]\label{def:subspace-design} A collection $\mathcal{H}$ of $\mathbb{F}$-linear subspaces $H_1,\dots,H_n\subseteq \mathbb{F}^k$ is called an $(\ell, A)$ subspace design over $\mathbb{F}$, if for every $\mathbb{F}$-linear space $W\subseteq \mathbb{F}^k$ of dimension $\ell$, we have
\[\sum_{i=1}^n\dim_{\mathbb{F}}\left(H_i\cap W\right)\leq A.\].
\end{definition}

\begin{definition}[\text{Subspace Designable Code, \cite[Definition B.2]{cz24}}]\label{def:designcodes} For any $s\ge 1$, given an $\mathbb{F}$-linear code $C\subseteq \left(\mathbb{F}^s\right)^n$ with message length $k$ and block length $n$, we use $\mathcal{C}\colon \mathbb{F}^k\to\left(\mathbb{F}^s\right)^n$ to denote the $\mathbb{F}$-linear encoder of $C$. For any $i\in[n]$, let $H_i\subseteq \mathbb{F}^k$ denote the $\mathbb{F}$-linear subspace such that for any message $f\in\mathbb{F}^k$, there is $\mathcal{C}(f)_i=0$ iff $f\in H_i$. We say $C$ is a $(\ell,A)$ subspace designable code if $\mathcal{H}:=\{H_1,\dots,H_n\}$ is an $(\ell, A)$ subspace design.
\end{definition}

We call an $s$-folded $\F$-linear code $\mathcal{C}\subseteq (\F^s_q)^n$ with rate $R$ \textbf{$d$-subspace designable} iff for all $1\leq d'\leq d$, $\mathcal{C}$ is a $(d',Rd'n+1)$ subspace designable code.

Since we will establish a negative bound \cref{thm:gk16-optimal-intro} later in this paper, and it indicates that for algebraically closed field, the quality of 
$d$-dimensional subspace design cannot be better than $\frac{d(k-d)}{s-d+1}\ge \frac{Rdsn}{s-d+1}-\frac{d^2}{s-d+1}=R(1+\frac{d-1}{s-d+1})dn-\frac{d^2}{s-d+1}$. For fixed $d\ge 1$ and large enough $n$, this lower bound is strictly larger than $Rdn$. Therefore, the $d$-subspace designable codes actually yield \textbf{optimal} subspace designs for subspaces with dimension at most $d$.

We also need a slacked version for ``nearly optimal'' subspace designable codes. Namely, for any real $\mu>0$ and integer $d\ge 1$, we say the code $\mathcal{C}\subseteq(\F^s)^n$ $\mu$-slacked $d$-subspace designable if for all $1\leq d'\leq d$, $\mathcal{C}$ is $(d',(R+\mu)d'n)$ subspace designable.
\subsubsection{Local Equivalence Between Random Linear Codes and Subspace Designs}

Our first main result is to build a local equivalence between random linear codes and nearly optimal subspace designable codes. Since explicit folded Reed--Solomon codes and univariate multiplicity codes yield nearly optimal subspace designs \cite{guruswami2016explicit}, our result gives the first explicit construction of folded linear codes that simultaneously preserves all local properties of random linear codes.

We will formally define the concept local profile later. For readers unfamiliar with this notion, a good intuition is to think a local profile as a certain list-decodability or list-recoverability of codes as special cases of our general reduction\footnote{Actually, list-decoding and list-recovery corresponds to a family of local profiles rather than a single one.}.

\begin{theorem}[Informal]\label{thm:main-equiv}
Fix $b,q\ge 1,\eps,R\in(0,1)$ and any 
$b$-local profile $\mathcal{V}=(V_1,\dots,V_n)\in\mathcal{L}(\F^b_q)^n$. The following holds.
\begin{compactitem}
\item \emph{(\cref{thm:FRS-threshold})}  If random $\F_q$-linear codes with rate $R$ does not contain $\mathcal{V}$ with high probability, then all $\F_q$-linear $\eps$-slacked $b$-subspace designable codes  with rate at most $R-\eps$ do not contain $\mathcal{V}$.
\item \emph{(\cref{thm:easy-direction})} For any $d\ge 1$, if for all large enough 
block length $n$, all $\F_q$-linear, $d$-subspace designable code with rate $R$ does not contain $\mathcal{V}$, then random $\F_q$-linear code with rate $R-o_n(1)-\eps$ does not contain $\mathcal{V}$ with probability at least $1-q^{-\eps n+b^2}$.
\end{compactitem}
\end{theorem}
As a direct corollary of \cref{thm:main-equiv}, we give a reduction from list-recoverability of optimal subspace designable codes to list-recoverability of random linear codes as stated in \cref{cor:sub-lr}. In a follow-up work \cite{BCDZ25lr}, the authors use this reduction to prove the nearly-optimal list-recoverability of random linear codes. See \cref{tab:previous} for a comparison with previous work.

\begin{table}
\begin{adjustbox}{width=\columnwidth,center}
\begin{tabular}{|c|c|c|c|}
  \hline
  Codes attain thresholds of RLC & List-decoding & List-recovery & All local properties \\
  \hline
  Random LDPC codes  & \cite{mrr20} & \cite{mrr20} & \cite{mrr20} \\
  \hline
  Randomly punctured low-bias codes  & \cite{gm22} & \cite{gm22} & \cite{gm22} \\
  \hline
  Random RS codes  & \cite{ST23,brakensiek2023generic,guo2023randomly,alrabiah2024randomly} & \cite{lms25} & \cite{lms25} \\
  \hline

  Explicit Constructions & \cite{ST23,bdg24b,cz24,jmst25} & This work & This work \\
  \hline
  \end{tabular}
  \end{adjustbox}
  \caption{This table compares various randomness-efficient families of linear codes as well as explicit codes that behave like random linear codes in terms of list-decoding, list-recovery, and the most general local properties. See also our discussion on the concurrent work of \cite{js25}.}
  \label{tab:previous}
  \end{table}

It was known from the pioneer work of Guruswami and Kopparty \cite{guruswami2016explicit} that when $s\ge b (R/\eps+1)$, the explicit $s$-folded Reed--Solomon codes and $s$-univariate multiplicity codes are $\eps$-slacked $b$-subspace designable codes as discussed later in \cref{rem:slacked-sub}.

Therefore, as a corollary, we give the first explicit construction of codes that simultaneously simulate all local properties of random linear codes. Concretely, it follows that
\begin{theorem}\label{thm:explicit-local}
Fix $b,q\ge 1,\eps,R\in(0,1),s\ge b(R/\eps+1)$, there are explicit $s$-folded Reed--Solomon codes and $s$-order univariate multiplicity codes $\mathcal{C}\subseteq(\F^s_q)^n$, where $q=\Theta(sn)$, with rate $R-\eps$, such that for any $b$-local profile $\mathcal{V}=(V_1,\dots,V_n)\in\mathcal{L}(\F^b_q)^n$, if random $\F_q$-linear code with rate $R$ does not contain $\mathcal{V}$ with high probability, then $\mathcal{C}$ does not contain $\mathcal{V}$. 

In particular, for any $\ell,L\ge 1$, if random $\F_q$-linear codes with rate $R$ are $(\rho,\ell,L)$-list recoverable, then $\mathcal{C}$ is $(\rho,\ell,L)$ list-recoverable.
\end{theorem}
Since it was previously shown by \cite{brakensiek2023generic,alrabiah2024randomly,lms25} that random linear codes achieve generalized singleton bound, \cref{thm:explicit-local} directly gives an alternative proof of the main theorem of \cite{cz24} that folded RS codes achieves generalized singleton bound \cite{ST23} as a special case.

Interestingly, since the main theorem of \cite{cz24} only uses the subspace design property (This generalization is presented in \cite[Appendix B]{cz24}), \cite{cz24} actually proves that all optimal subspace designable codes achieve generalized singleton bound. Therefore, combined with \cite{cz24}, the second item of \cref{thm:main-equiv} actually conversely provides an alternative proof that random linear codes achieve generalized singleton bounds originally by \cite{brakensiek2023generic,alrabiah2024randomly,lms25}.
\begin{remark}[List-recoverability]
Fix radius $1-R-\eps$ and input list size $\ell\ge 2$, the best known $(1-R-\eps,\ell,L)$ list-recoverability of random linear codes is $L=(\ell/\eps)^{O(\ell/\eps)}$ by \cite{ls25}. As an algebraic miracle, the bound for explicit folded RS codes is $L=(\ell/\eps)^{O(1+\log{(\ell)}/\eps)}$ by \cite{kopparty2023improved,tamo24}, which is better! As we have seen before that \cite{cz24} plus our reduction implies the same bound for random linear codes \cite{brakensiek2023generic,alrabiah2024randomly,lms25}, it is tempting to use \cref{thm:main-equiv} on the bound in \cite{kopparty2023improved,tamo24} to improve the bound of \cite{ls25}. Regrettably, this cannot work. The reason is that, unlike \cite{cz24}, the proof of \cite{kopparty2023improved,tamo24} is not ``subspace-design-only''. Concretely, their proof intrinsically relies on the fact that bad list resides in some low-dimensional affine subspace, which was proved in \cite{guruswami2013linear} using an interpolation technique that is tailored for algebraic codes, This does not apply to merely optimal subspace designable codes and we cannot apply \cref{thm:main-equiv} on their bound.

However, in a follow-up work \cite{BCDZ25lr}, the authors prove a better bound $(\frac{\ell}{R+\eps})^{O(R/\eps)}$ in a ``purely subspace-design'' way. This directly improves the bound of \cite{kopparty2023improved,tamo24} for folded RS codes. Moreover, this ``subspace design nature'' of their proof unlocks the full power of \cref{thm:main-equiv} and enables them to use this reduction to transfer the same bound directly to random linear codes, which improves the bound of \cite{ls25}.
\end{remark}

\paragraph{Concurrent Work.}Concurrently and independently, in an exciting work by Jeronimo and Shagrithaya \cite{js25}, the authors use the AEL-type constructions 
\cite{ael95} to get explicit constructions of folded linear codes 
over \textbf{constant-sized} alphabet that simulates any fixed ``\textbf{reasonable} local property''. We comment and compare their results with ours as follows.
\begin{compactitem}
\item For any fixed family of $b$-local profiles $\mathcal{V}_1,\dots,\mathcal{V}_m$, if each of them has at most $T$ distinct linear subspaces within it, then \cite{js25} constructs explicit codes with alphabet $q^{\Theta(T^3)}$  that simulates all these local profiles  when $m=q^{o(n)}$. However, for the most general family of local profiles, these parameters could be $m=q^{\Theta(b^2n)}$ and $T=\min(q^{\Theta(b^2)},n)$. Although this makes \cite{js25} inaccessible in the most general setting, it is not an issue when $q$ is a large constant independent of $n$, which is the primary case of interest in \cite{js25}. By contrary, our construction uses a single algebraic code to simulate all possible $b$-local profiles simultaneously.

\item As the cost of simultaneous simulation of all local profiles, our explicit construction requires alphabet size $q=\Theta(bn)$, while the construction of \cite{js25} requires only constant alphabet when the previous conditions are met, which is better in these cases. We further comment that in \cref{thm:random-subspace-design} we prove that there exists nearly optimal subspace designs over constant alphabet. Therefore, future improvement on explicit constructions of subspace designs over smaller fields would reduce the alphabet size required.
\item Before our work, local properties are primarily interested for the reasons of list-decoding and list-recovery of linear codes. For these two major applications, the conditions in the first item are met so both \cite{js25} and our results can handle list-decoding and list-recovery when the input and output list sizes are constants. However, in our second part of results relating to tensor codes, $T$ and 
$q$ could be enormous so we have to simulate local profiles in the most general setting.
\end{compactitem}

\subsubsection{Applications to Tensor Codes and Matroid Theory}\label{subsec:matroid-results}

As mentioned earlier, two commonly study local properties correspond to the list-decodability and list-recoverability of linear codes. A novel conceptual contribution we make in this paper is that the theory of local properties (with slight modification) can be used to study the \emph{tensoring} of linear codes. Given a column code $C_{\col} \subseteq \F_q^m$ and a row code $C_{\row} \subseteq \F_q^n$, we define their tensor product $C_{\col} \otimes C_{\row} \subseteq \F_q^{m\times n}$ to be the set of all $m \times n$ matrices such that each column is a codeword in $C_{\col}$ and each row is a codeword in $C_{\row}$. In particular, $\dim(C_{\col} \otimes C_{\row}) = \dim(C_{\col}) \dim(C_{\row})$.

As a fundamental linear algebraic operation, tensor products are a key operation in many constructions within coding theory as well as more broadly in TCS~\cite{Gopalan2016,panteleev2022asymptotically,dinur2022locally,ccs25,kalachev2025Maximallya,berczi2025matroid}. In this paper, we use the theory of local properties to understand better the \emph{correctability} of tensor products with respect to erasures. That is, given an erasure pattern $E \subseteq [m] \times [n]$, when can we recover a codeword $c \in C_{\col} \otimes C_{\row}$ when only given the symbols $c|_{\bar{E}}$? This is equivalent to asking whether there exists a codeword $c \in C_{\col} \otimes C_{\row}$ supported on $E$. Of course, the answer depends on the precise structure on the precise structure of $C_{\col}$ and $C_{\row}$, but there is much interest in understanding best-case scenario. More formally, Gopalan et al.~\cite{Gopalan2016} asked the following.

\begin{question}\label{ques:MR}
Given parameters $m,n,a,b$ as well as a pattern $E \subseteq [m] \times [n]$. When does there exist (for sufficiently large $q$) a $m-a$-dimensional column code $C_{\col} \subseteq \F_q^m$ and a $n-b$-dimensional row code $C_{\row} \subseteq \F_q^n$ such that $E$ is a correctable pattern in $C_{\col} \otimes C_{\row}$?
\end{question}

In the case $a=1$, Gopalan et al.~\cite{Gopalan2016} gave a combinatorial description, which was later turned into a deterministic\footnote{We note that \Cref{ques:MR} lies in $\mathsf{RP}$ by picking $C_{\col}$ and $C_{\row}$ to be random linear codes and invoking Schwarz-Zippel~\cite{Sch80,Zip79}. The core challenge is to \emph{derandomize} this special case of Polynomial Identity Testing (PIT).} polynomial-time checkable condition by Brakensiek, Gopi, and Makam~\cite{bgm2021mds,brakensiek2023generic}. However, the approach used in those works cannot extend to $a \ge 2$~\cite{holzbaur2021correctable}. Although \Cref{ques:MR} may seem like a niche coding theory question on the surface, a recent characterization by Brakensiek et al.~\cite{brakensiek2024Rigidity} shows that answering \cref{ques:MR} is equivalent to resolving some long-standing questions in \emph{matroid} theory, including understanding the structural rigidity of bipartite graphs~\cite{kalai2015Bipartite} as well as when partially revealed matrices can be completed to low-rank matrices~\cite{singer2010uniqueness}. See \Cref{sec:matroid} as well as the recent survey by Cruickshank et al.~\cite{cruickshank2025rigidity} for a more thorough discussion of these connections.

We give a novel plan of attack toward resolving \Cref{ques:MR} by studying correctability in a tensor code as a a slight variant of a local property in the sense of \cite{lms25}. By adapting a suitable potential function designed in \cite{lms25}, (see \Cref{subsec:potential-matroid}) we give a one-sided test for checking whether $E$ is a correctable pattern in \Cref{ques:MR}. By ``one-sided,'' we mean than the potential test might report than $E$ is correctable when $E$ is actually uncorrectable, but the opposite can never occur. By suitably adapting \Cref{thm:main-equiv} (see \Cref{subsec:pass-to-FRS}), we can replace $C_{\col}$ and $C_{\row}$ with two explicit folded Reed-Solomon codes (with the caveat the folding parameters are $\poly(m,n)$-sized) and use those to deterministically check the correctability conditions.

To understand the accuracy of this algorithm, one needs to know whether there exists patterns $E$ which are uncorrectable in \Cref{ques:MR}, but are correctable according to the folded Reed-Solomon codes. We connect this question to a long-standing question\footnote{This question on the tensor products of matroids should not be confused with Mason's log-concavity conjectures which were recently resolved~\cite{adiprasito2018hodge,huh2021correlation,branden2020lorentzian,huh2022combinatorics,anari2024log}.} by Mason~\cite{mason1981glueing} on the structure of the tensor products of matroids. Informally, Mason asks whether the linear algebraic tensor product of two random linear codes (over a large field) is the most general matroid which ``looks like'' a tensor product. See \Cref{conj:birigid} for the precise assumption we make, which is a modern reformulation oif Mason's question by Cruickshank et al.~\cite{cruickshank2025rigidity} which is more convenient for our purposes. With this assumption, we can prove that the potential function we construction is a \emph{zero-error} test for correctability (see \Cref{cor:potential-MR} and \Cref{lem:tensor-two}). As such, we prove the following result.

\begin{theorem}[Informal, see \Cref{thm:MR-algo}]\label{thm:main-matroid}
Assuming a positive answer to a question of \emph{Mason~\cite{mason1981glueing}}, \Cref{ques:MR} can be answered in deterministic polynomial time.
\end{theorem}

This result improves over a recent result of Jackson and Tanigawa~\cite{jackson2024maximal} who conditionally showed than \Cref{ques:MR} lies in $\mathsf{RP} \cap \mathsf{coNP}$. We note however their conditional assumption was different (and stronger), see \Cref{subsec:abstract} for further discussion. See also \Cref{rem:conj-pro-con} for a thorough discussion on evidence for and against \Cref{conj:birigid}.

The recent work of B{\'e}rczi et al.~\cite{berczi2025matroid} also studies products between matroids, although they are interested in constructing products for an arbitrary pair of matroids, whereas we are interested in studying more deeply the tensor products of linear matroids.

\subsubsection{Tight Bounds for the Quality of Subspace Designs}

Given the aforementioned applications of subspace designable codes, a natural topic to investigate is the optimal parameters of explicit subspace designs. In Section~\ref{sec:subspace-design}, we make new progress toward answering such questions. As a baseline, we first present parameters of the explicit construction of subspace designable codes in \cite{guruswami2016explicit}.
\begin{theorem}[\cite{guruswami2016explicit}, as stated in \cite{cz24}]\label{thm:code-construction-sub} For any $s,n\ge 1,q\ge sn\ge k$, there are explicit constructions of  $\F_q$-linear codes $\mathcal{C}\subseteq (\F^s_q)^n$ such that for all $d\in\{0,1,\dots,s\}$, $\mathcal{C}$ is a $(d,\frac{d(k-1)}{s-d+1})$ subspace designable code. In particular, explicit folded RS codes and univariate multiplicity codes are such constructions.
\end{theorem}
\begin{remark}\label{rem:slacked-sub}
This parameter implies that, for any $b\ge1,\eps\in(0,1)$ and rate $R=k/(sn)$, when $s\ge b(R/\eps+1)$, explicit $s$-folded Reed--Solomon codes and $s$-order univariate multiplicity codes are $\eps$-slacked $b$-subspace designable since we can calculate $d(k-1)/(s-d+1)\leq (R+\eps)dn$ for all $d\in[b]$ under this choice of $s$. 
\end{remark}

Our first contribution is a more fine-grained analysis of \cref{thm:code-construction-sub} that slightly improves its parameters.

\begin{theorem}[Informal, see \Cref{thm:gk16-improved}]\label{thm:gk16-improved-intro}
Assume that $q > ns$, then there are explicit constructions of  $\F_q$-linear codes $\mathcal{C}\subseteq (\F^s_q)^n$ such that for all $d\in\{0,1,\dots,s\}$, $\mathcal{C}$ is a $(d,\frac{d(k-d)}{s-d+1})$ subspace designable code. In particular, explicit folded RS codes are such constructions.
\end{theorem}

The proof of Theorem~\ref{thm:gk16-improved-intro} largely follows the proof of \cref{thm:code-construction-sub}, except we correct one inefficiency in the root-counting bound of \cite{guruswami2016explicit} which leads to the slightly tighter parameters. More surprisingly, we show that the bound is \Cref{thm:gk16-improved-intro} is \emph{exactly} tight, as long as the base field $\F$ is algebraically closed (e.g., $\mathbb C$ or the algebraic closure of $\F_2$).

\begin{theorem}[Informal, see \Cref{thm:gk16-optimal}]\label{thm:gk16-optimal-intro}
Let $\F$ be an algebraically closed field (of any characteristic). Consider any $\F$-linear codes $\mathcal{C}\subseteq (\F^s)^n$  and parameter $d \in \{1,\dots, s\}$. If $n \le \frac{d(k-d)}{s-d+1}$ then $\mathcal C$ is \emph{not} a $(d, \frac{d(k-d)}{s-d+1}-1)$ subspace designable code.
\end{theorem}

In particular, any improvements over \Cref{thm:gk16-improved-intro} must crucially use the fact that $\F_q$ fails to be algebraically closed. In \Cref{sec:counterexample}, we give a simple example over $\F_3$ of a (constant-sized) subspace design utilizing this property to beat \Cref{thm:gk16-improved-intro}.

\begin{remark}
A recent work by Santonastaso and Zullo~\cite{santonastaso2023subspace} also studies the optimal parameters for subspace designs over finite fields, with the motivation of achieving a better understanding of the parameters sum-rank codes. However, the regime they work in is rather different from ours. In particular, they are primarily interested in the regime for which $d = k-1$, although they allow the test space $U$ to be a vector space over an extension field. They also allow $\dim(U) + \dim(H_i) > k$, which we do not allow. In the most comparable regime of $s = d = k-1$, they get that $\mathcal C$ is not $(k-1, k-2)$ subspace designable over any field, which matches our bound for algebraically closed fields. Of note, \Cref{thm:code-construction-sub} by \cite{guruswami2016explicit} constructs a $(k-1, (k-1)^2)$ design in this regime, which is improved to the tight bound of $(k-1, k-1)$ by our \Cref{thm:gk16-improved-intro}.
\end{remark}

The proof of \Cref{thm:gk16-optimal-intro} makes nontrivial use of algebraic geometry. More precisely, we observe that the space of $d$-dimensional subspaces $W \subseteq \F^k$ can be viewed as the $d(k-d)$-dimensional \emph{Grassmannian} projective variety (see \cite{lakshmibai2015Grassmannian}). A subspace design constraint of the form $\dim(H_i \cap W) \ge 1$, carves out a what is known as a \emph{Schubert subvariety} of the Grassmannian. The key property of this Schubert subvariety is that it has codimension $s-d+1$. A key property of subvarieties over algebraically closed fields is that the codimension of the intersection of subvarieties is a a subadditive function of the codimensions of the respective subvarieties~(e.g., \cite{hartshorne1977Algebraic}). Thus, if we seek to rule out all $W$ in the Grassmanian as counterexamples to \Cref{def:subspace-design}, we need to intersect at least $\lfloor\frac{d(k-d)}{s-d+1}\rfloor + 1$ Schubert subvarieties, implying the subspace design parameter shouild be at least $\lfloor\frac{d(k-d)}{s-d+1}\rfloor > \frac{d(k-d)}{s-d+1}-1$.

\subsection*{Open Questions}

Given the many connections between our work and various parts of mathematics and computer science, there are a variety of directions to pursue.

\begin{itemize}
\item \Cref{thm:random-subspace-design} shows that near-optimal subspace designable codes exist over constant-sized fields. Can such codes be explicitly constructed?
\item Conversely, can \Cref{thm:gk16-optimal-intro} be extended to the setting of finite fields? Recall we show in \Cref{sec:counterexample} the the lower bound must to degrade to some extent, but is a lower bound of $(1-o(1))\frac{d(k-d)}{s-d+1}$ possible?
\item From our connection between local properties and rigidity matroids, a key open question is to prove or disprove \Cref{conj:birigid}. Either outcome has interesting ramifications for both coding theory and rigidity theory, see \Cref{rem:conj-pro-con} for further discussion.
\item We also emphasize the long-standing open question (since the 1800s~\cite{Maxwell,graver1993combinatorial}) of giving a deterministic polynomial time algorithm for detecting (non-bipartite) rigid graphs in $\R^3$. The methods used in this paper likely are insufficient to resolve this question, but further understanding on derandomizing (not necessarily local) properties of random linear codes could be critical to its resolution.
\end{itemize}

\subsection*{Organization}

In \Cref{sec:prelim}, we present some basic notation as well as discuss the theory of local properties by Levi, Mosheiff, and Shagrithaya~\cite{lms25}. In \Cref{sec:random-to-explicit}, we show that subspace design codes of slacked rate capture all local properties of random linear codes, establishing the first half of \Cref{thm:main-equiv}. In \Cref{sec:subspace-to-random}, we show that local properties shared by all subspace design codes also apply to random linear codes, proving the second half of \Cref{thm:main-equiv}. In \Cref{sec:matroid}, we adapt the theory of local properties to the study of quesitons in matroid theory, proving \Cref{thm:main-matroid}. In \Cref{sec:subspace-design}, we study the optimal parameters of subspace designs, proving \Cref{thm:gk16-optimal-intro} and \Cref{thm:gk16-improved-intro}.

\section{Preliminaries}\label{sec:prelim}

\paragraph{Notation.} We describe some basic notation. For a matrix $M\in \Sigma^{n\times m}$, we let $M[i,j]\in\Sigma,i\in[n],j\in[m]$ denote the entry at the $i$-th row and $j$-th column. We use $M[i:]\in\Sigma^m$ to denote the $i$-th row of $M$ and $M[:j]\in\Sigma^n$ to denote the $j$-th column of it. Assume $U,V$ are two linear spaces over field $\F$. We use $\dim(U)$ to denote $\F$-dimension of $U$. For a set of vectors $f_1,\dots f_m\in U$, we use $\Span(f_1,\dots,f_m)$ to denote the linear subspace of $U$ $\F$-spanned by $f_1,\dots,f_m$. Given a linear map $\psi\colon U\to V$, we let $\ker(\psi)\colon=\{f\in U\colon\psi(f)=0\}$ denote the kernel of $\psi$. By the first isomorphism theorem, $\ker(\psi)$ is an $\F$-linear subspace of $U$. 

\begin{definition}[Folded Wronskian, see \cite{guruswami2016explicit}]\label{def:wronskian} Let $f_1(X), \ldots, f_s(X) \in \mathbb{F}_q[X]$ and $\gamma \in \mathbb{F}_q^\times$. We define their $\gamma$-folded Wronskian $W_\gamma\left(f_1, \ldots, f_s\right)(X) \in \left(\mathbb{F}_q[X]\right)^{s\times s}$ by
$$
W_\gamma\left(f_1, \ldots, f_s\right)(X) \stackrel{\text { def }}{=}\left(\begin{array}{ccc}
f_1(X) & \ldots & f_s(X) \\
f_1(\gamma X) & \cdots & f_s(\gamma X) \\
\vdots & \ddots & \vdots \\
f_1\left(\gamma^{s-1} X\right) & \cdots & f_s\left(\gamma^{s-1} X\right)
\end{array}\right).
$$  
\end{definition}

It is well known that the nonsingularity of the folded Wronskian of a set of vectors characterizes the linear independence of these vectors, as stated below.
\begin{lemma}[Folded Wronskian criterion for linear independence, see \cite{guruswami2016explicit,guruswami2013linear}]\label{lem:wrons_ind} Let $k<q$ and $\vec{f_1}, \ldots, \vec{f_s} \in \mathbb{F}^k_q$. 
Let $\gamma$ be a generator of $\mathbb{F}_q^\times$.
Then $\vec{f_1}, \ldots, \vec{f_s}$ are linearly independent over $\mathbb{F}_q$ if and only if the folded Wronskian determinant $\operatorname{det} W_\gamma\left(f_1, \ldots, f_s\right)(X) \neq 0$.
\end{lemma}

\subsection{Local Profiles and Thresholds of Random Linear Codes}
In this section, we describe the work Levi, Mosheiff, and Shagrithaya \cite{lms25} on the theory of local profiles for random linear codes. We closely follow their notation.

Fix a finite field $\F_q$, block length $n$ and configuration parameter $b\ge1$. Let $\mathcal{L}(\mathbb{F}^b_q)$ denote the set of all linear subspaces of $\mathbb{F}^b_q$. Let $\Ldist(\F^b_q)$ denote the set of all linear subspaces $U$ of $\F^b_q$  whose basis matrix $B_U\in\F^{b\times\dim(U)}_q$ (This means columns of $B_U$ form a basis of $U$) does not contain repeated rows. This is equivalent to say that for any $i\neq j\in[b]$, there exists some $f\in U$ such that $f_i\neq f_j$.

We call $\mathcal{V}=(V_1,\dots,V_n)\in\mathcal{L}(\mathbb{F}^b_q)^n$ a 
$b$-local profile. For any 
$b$-local profile $\mathcal{V}$ and linear subspace $U\in \mathcal{L}(\mathbb{F}^b_q)$, \cite{lms25} use $\mathcal{M}_{\mathcal{V},U}$ to denote the set of matrices $M\in \mathbb{F}^{n\times b}_q$ such that 
\begin{itemize}
\item[(1)] $M$ has pairwise distinct columns. \item[(2)] for any $i\in[n]$, the $i$-th row $M_i$ of $M$ satisfies $M_i\in V_i\cap U$. 
\item[(3)] The row span of $M$ is $U$.
\end{itemize}

They use $\cM^*_{\cV,U}$ to denote the set of matrices that satisfy condition (2). For any $\mathbb{F}_q$-linear code $\mathcal{C}\subseteq \mathbb{F}^n_q$, we say $\mathcal{C}$ contains $(\mathcal{V},U)$ if there is some $M\in \mathcal{M}_{\mathcal{V},U}$ such that each column of $M$ is in $\mathcal{C}$. We say $\mathcal{C}$ contains $\mathcal{V}$ if there is some $U\in \mathcal{L} (\mathbb{F}^b_q)$ such that $\mathcal{C}$ contains $(\mathcal{V},U)$. \cite{lms25} defines the potential function of $(\mathcal{V},U)$ as
\begin{equation}\label{eq:potential}
\Phi(\mathcal{V},U,R)=-n\dim(U)+\sum^n_{i=1}\dim(V_i\cap U)+Rn\dim(U).
\end{equation}
Moreover, \cite{lms25} defines  $R_{\mathcal{V}}$ of a $b$-local profile $\mathcal{V}$ as follows.
\[
R_{\cV}=\max\{R\in[0,1]\colon \forall U\in \Ldist(\F^b_q),\exists W\subsetneq U,\text{ s.t. }\Phi(\cV,U,R)-\Phi(\cV,W,R)\leq 0\}.
\]
We observe the following simple fact about $R_{\cV}$.
\begin{fact}\label{fac:monotone}
For any $b$-local profile $\mathcal{V}$, rate $R\in[0,1]$, and subspace $U\in\Ldist(\F^b_q)$, if $R\leq R_{\mathcal{V}}$, there exists a proper subspace $W\subsetneq U$ such that $\Phi(\mathcal{V},U,R)-\Phi(\mathcal{V},W,R)\leq 0$.
\end{fact}
\begin{proof}
It follows from the definition of $R_{\mathcal{V}}$ and the fact that for any $W\subsetneq U$, $\frac{\mathsf{d}(\Phi(\mathcal{V},U,R)-\Phi(\mathcal{V},W,R))}{\mathsf{d}R}=n(\dim(U)-\dim(W))\ge 0$.
\end{proof}
One of the main contributions of \cite{lms25} is that they show $R_{\mathcal{V}}$ is exactly the rate threshold for random $\F_q$-linear code to contain $\mathcal{V}$, which we quote as follows.
\begin{theorem}[{\cite[Theorem 4.4]{lms25}}]\label{thm:threshold}
Let $\mathcal{V}=(V_1,\dots,V_n)\in\mathcal{L}(\F^b_q)^n$ be a $b$-local profile. Let $\mathcal{C}\subseteq \F^n_q$ be a random $\F_q$-linear code of rate $R$. The following holds.
\begin{itemize}
\item[(1)] If $R\ge R_{\mathcal{V}}+\eps$, then $\Pr[\mathcal{C} \text{ contains }\mathcal{V}]\ge 1-q^{-\eps n+b^2}$.
\item[(2)] If $R\leq R_{\mathcal{V}}-\eps$, then $\Pr[\mathcal{C} \text{ contains }\mathcal{V}]\leq q^{-\eps n+b^2}$.
\end{itemize}
\end{theorem}

\section{All Subspace Designable Codes Locally Simulate RLCs}\label{sec:random-to-explicit}

Fix $s,b,n\ge 1$ and finite field $\F_q$. For any $s$-folded $\F_q$-linear code $\mathcal{C}\subseteq (\F^s_q)^n$ and $b$-local profile $\mathcal{V}=(V_1,\dots,V_n)\in\mathcal{L}(\mathbb{F}^b_q)^n$, let $\mathcal{C}'\subseteq \F^{sn}_q$ denote the ``unfolded'' $\F_q$-linear code $\mathcal{C}$ with block length $sn$ and define the $s$-duplicated $b$-local profile $\mathcal{V}^{(s)}:=(V_1,\dots,V_1,V_2,\dots,V_2,\dots,V_n,\dots,V_n)\in\mathcal{L}(\F^b_q)^{sn}$ to be the $b$-local profile with length $sn$ derived from $\mathcal{V}$ such that each $V_i$ is repeated $s$ times.  We say $\mathcal{C}$ contains $\mathcal{V}$ iff $\mathcal{C}'$ contains $\mathcal{V}^{(s)}$.

The main result of this section is that nearly-optimal subspace designable codes simultaneously avoid all local profiles that a random linear code could avoid with high probability, with a cost of an arbitrarily small rate decrease. This establishes the harder direction of the local equivalence between subspace desgins and random linear codes. 

For any $b$-local profile $\mathcal{V}=(V_1,\dots,V_n)\in\mathcal{L}(\F^b_q)^n$, from \cref{thm:threshold} we know that random $\F_q$-linear code with rate arbitrarily close to $R_{\mathcal{V}}$ could avoid $\mathcal{V}$ with high probability. Therefore, our target is the following.
\begin{theorem}\label{thm:FRS-threshold}
Let $s,b\ge 1,\mu>0$ and $\mathcal{C}\subseteq (\F^s_q)^n$ be a $\mu$-slacked $b$-subspace designable code with rate $R=k/(sn)$. For any $b$-local profile $\mathcal{V}=(V_1,\dots,V_n)\in\mathcal{L}(\F^b_q)^n$, if $R\leq R_{\mathcal{V}}-\mu-1/n$, then $\mathcal{C}$ does not contain $\mathcal{V}$.  
\end{theorem}
\paragraph{Technical Overview.} Our proof is initially inspired by the similarity between the definition of potential function (\ref{eq:potential}) and the defining inequaliltes of subspace designs. When the rate $R<R_{\mathcal{V}}$ is smaller than the threshold, we know that for any ``coordinate configuration'' linear subspace $U\subseteq \F^b_q$, it must have a proper linear sbuspace $W$ such that $\Phi(U,\mathcal{V},R)< \Phi(W,\mathcal{V},R)$. In previous work on local properties \cite{lms25}, the main proof strategy was to exclude the possibility that the ``bad list of codewords'' could span $U$ using their ``vectors on coordinates'' simultaneously for all $U$. They use the critical subspace $W\subsetneq U$ to project the original local profile $\mathcal{V}$ onto $W^{\bot}$ to get a new local profile $\mathcal{V}'$, and then reduce the original problem to the smaller linear subspace $W^{\bot}$ with respect to $\mathcal{V}'$, and calculate the potentials for this case and complete the proof.

However, if we wish to establish the result from subspace designs, we have to find a linear space $N\subseteq \F^k_q$ in the \textbf{message space}, which resides in $\F^k_q$, to trigger the defining inequality of subspace designs. As discussed above, \cite{lms25} primarily works over the \textbf{coordinate configuration} linear subspace $U\subseteq \F^b_q$, which resides in $\F^b_q$. Therefore, it was not clear how to find the ``working linear subspace'' in the message space, so we need a different proof strategy. Unlike \cite{lms25} that simultaneously\footnote{In the random RS code setting, \cite{lms25} starts from the single ``whole coordinate space'' $\F^b$, but it is still not in the message space.} avoids all ``coordinate subspace'' $U\subseteq \F^b_q$, we would only work on a specific $U\subseteq \F^b_q$ depending on any potential ``bad list of messages''  that we would like to exclude. Concretely, suppose for the sake of contradiction there exists some ``bad list of messages'' that certifies $\mathcal{V}$, we will
define an ``associated coordinate subspace'' $U\subseteq \F^b_q$ generated from the bad list. Guaranteed by the definition of the threshold rate $R_{\mathcal{V}}$, there exists a proper linear subspace $W\subsetneq U\subseteq\F^b_q$ with larger potential. Then, from $U$ and $W$, we can construct some linear subspace $N\subseteq \F^k_q$ in the message space. Finally, if we calculate the defining inequality of subspace designs on $N$, it will lead to a contradiction and concludes that the bad list of messages does not exist. Our main technical novelty is the way to find the specific $U\subseteq\F^b_q$ and $N\subseteq \F^k_q$ such that the  potential functions $\Phi(U,\mathcal{V},R)-\Phi(W,\mathcal{V},R)$ can be perfectly transformed to ``subspace intersections'' between $N\subseteq \F^k_q$ and the subspace design.

\begin{proof}[Proof of \cref{thm:FRS-threshold}]
Let $\pi_1,\dots,\pi_n\colon \F^k_q\to\F^s_q$ be linear maps corresponding to $\mathcal{C}$ such that for any $i\in[n],f\in\F^k_q$, there is $\pi_i(f)=\mathcal{C}(f)[i]\in\F^s_q$. Since $\mathcal{C}$ is a $\mu$-slacked $b$-subspace designable code, we know that for any $1\leq d'\leq d$, $\ker(\pi_1),\dots,\ker(\pi_n)\subseteq\F^k_q$ form a $(d',(R+\mu )d'n)$ subspace design. 

Suppose by contrapositive $\cC$ contains $\cV$,  we have distinct  $f_1,\dots,f_b\in\F^k_q$ such that for any $i\in[n]$ and $j\in[s]$, there is $(\pi_i(f_1)[j],\dots,\pi_i(f_b)[j])\in V_i$. At least one of $f_1,\dots,f_b$ is non-zero.  Let $F=\Span(f_1,\dots,f_b)$ and $\dim F=t\ge 1$. Let $g_1,\dots,g_t\in\F^k_q$ be an arbitrary basis of $F$, we define a matrix $K\in\F^{t\times k}_q$ such that the $i$-th row of $K$ is $g_i$ for each $i\in[t]$. Let $\varphi_F\colon \F^t_q\to F$ denote the linear map defined as $\varphi_F(v)=K^Tv,\forall v\in\F^t_q$. Since $\rank(K)=t$, $\varphi_F$ is actually an isomorphism from $\F^t_q$ to $F$. For each $f_i,i\in[b]$, there is a unique way to write $f_i=\sum^t_{j=1}m_{ij}g_j$ as a $\F_q$-linear combination of the basis $g_1,\dots,g_t$.

Let $M=(m_{ij})\in\F^{b\times t}_q$ be the matrix corresponding these coefficients. Since $g_1,\dots,g_t$ form a basis, so the linear subspace spanned by rows of $M$ is exactly $F$, which implies $\rank(M)=t$. Also, since $\vec{g}_1,\dots,\vec{g}_t$ are linear independent and $\vec{f}_1,\dots,\vec{f}_{b}$ are distinct, it is clear that $M$ has distinct rows. Therefore, let $U\subseteq \F^b_q$ denote the linear subspace spanned by columns of $M$. Namely, let $d_1,\dots,d_t\in\F^b_q$ the columns of $M$ such that $d_i=M[:i],i\in[t]$, we know that $U=\Span(d_1,\dots,d_t)$. We claim that $U\in\Ldist(\F^b_q)$.  To show this, for any $u\neq v\in[b]$, since $M$ has distinct rows, $M[u:]\neq M[v:]$. This means for some $w\in[t]$ there is $M[u,w]\neq M[v,w]$, so the $w$-th column $M[:w]=d_w\in U$ satisfies that $d_{w,u}\neq d_{w,v}$. Therefore,   we know $U\in \Ldist(\F^b_q)$. Since $\rank(M)=\dim(U)=t$, we know $d_1,\dots,d_t$ form a basis of $U$. We define the linear map $\varphi\colon \F^b_q\to\F^t_q$ as $\varphi(v)=M^{\top}v$ for all $v\in\F^b_q$. Since $\rank(M)=t$, $\varphi$ is surjective (epimorphism).

Then, let $R'=R+\mu+1/n$, from \cref{fac:monotone} we know there is some proper linear subspace $W\subsetneq U$ such that
\[
\Phi(\cV,U,R')-\Phi(\cV,W,R')\leq 0.
\]

Let $m=\dim(W)\ge 0$. Note that $m<t$ since $W$ is a proper linear subspace of $U$.  There must be some surjective linear map (epimorphism) $\psi\colon\F^{\top}_q\to\F^m_q$ such that  $\ker(\psi\circ\varphi)\subseteq \F^b_q$ is exactly $W^{\bot}$. To show the existence of $\psi$, there are two cases. If $m=0$, then $\psi$ must be trivial and we are done. Otherwise, let $w_1,\dots,w_m\in W\subsetneq U\subseteq \F^b_q$ denote a basis of $W$. There is a unique way to write $w_i=\sum^t_{j=1}h_{ij}d_j$ as an $\F_q$-linear combination for each $i\in[m]$. Let $H=(h_{ij})\in\F^{m\times t}_q$, the linear map $\psi\colon \F^t_q\to\F^m_q$ defined as $\psi(v)=Hv,\forall v\in\F^t_q$ satisfies the requirement since $(HM^{\top})[i]=w_i,\forall i\in[m]$, which means $\ker(\psi\circ\varphi)=\Span(w_1,\dots,w_m)^{\bot}=W^{\bot}$. $\psi$ is surjective since the rows of $H$ are  $w_1,\dots,w_m$, which are linearly independent.

Let $N=\ker(\psi)\subseteq \F^t_q$, we know $\dim(N)=t-m\ge 1$ since $\psi$ is surjective. we consider the linear space $N'=\varphi_F(N)\subseteq F\subseteq \F^k_q$. Since $\varphi_F$ is injective, it follows that $\dim(N')=t-m\leq b$.

Recall that $\mathcal{V}=(V_1,\dots,V_n)\in\mathcal{L}(\F^b_q)^n$ is the interested $b$-local profile, we have the following claim
\begin{claim}\label{clm:blackbox}
For any $i\in[n]$, there is $\varphi_F(\varphi(V^{\bot}_i)\cap N)\subseteq N'\cap \ker(\pi_i)$.
\end{claim}
\begin{proof}
Fix any $i\in[n]$, our proof is splitted into two pieces.
\begin{compactitem}
\item[(1)] $\varphi_F(\varphi(V^{\bot}_i)\cap N)\subseteq N'$: Since $N'=\varphi_F(N)$, this part is trivial.
\item[(2)] $\varphi_F(\varphi(V^{\bot}_i)\cap N)\subseteq \ker(\pi_i)$: For any $v\in V^{\bot}_i$, it suffices to prove $(\pi_i\circ \varphi_F\circ \varphi)(v)=0$. By the definition of $\pi_i$, let $H_i\in\F^{s\times k}_q$ denote the matrix such that $\pi_i(f)=H_if,\forall f\in\F^k_q$. Therefore, the target is equivalent to  $H_iK^{\top}M^{\top}v=0$. It suffices to show that for any $j\in[s]$, the $j$-th row $z_j\in\F^b_q$ of $H_iK^{\top}M^{\top}\in V_i$. Actually, since $j$-th row of $H_iK^{\top}$ is $(\pi_i(g_1)[j],\dots,\pi_i(g_t)[j])\in\F^t_q$ by definition, and $f_u=\sum^t_{v=1}M[u,v]g_v,\forall u\in[b]$, by $\F_q$-linearity of $\mathcal{C}$, we know that
\begin{align*}
z_j&=(\sum^t_{v=1}\pi_i(g_v)[j]M^{\top}[v,1],\dots,\sum^t_{v=1}]\pi_i(g_v)[j]M[v,b])\\
&=\left(\pi_i\left(\sum^t_{v=1}g_vM[1,v]\right)[j],\dots,(\pi_i\left(\sum^t_{v=1}g_vM[b,v]\right)[j]\right)\\
&=(\pi_i(f_1)[j],\dots,\pi_i(f_b)[j])\in V_i.
\end{align*}
This completes the proof of the claim.
\end{compactitem}
\end{proof}

Since $\varphi_F$ is injective. $\dim(\varphi_F(\varphi(V^{\bot}_i)\cap N))=\dim(\varphi(V^{\bot}_i)\cap N)$. By \cref{clm:blackbox}, it follows that $\dim(N'\cap \ker(\pi_i))\ge\dim(\varphi(V^{\bot}_i)\cap N)$. Since $\dim(N')\leq b$ and $\pi_1,\dots,\pi_n$ are $(\dim(N'),(R+\mu)\dim(N')n)$-subspace designable, it follows that
\begin{equation}\label{eq:sub-des}
(R+\mu)\dim(N')n\ge\sum^n_{i=1}\dim(N'\cap\ker(\pi_i))\ge \sum^n_{i=1} \dim(\varphi(V^{\bot}_i)\cap N)
\end{equation}

Let us calculate $\dim(\varphi(V^{\bot}_i)\cap N)$. Since $N=\ker(\psi)$, we have,
\[
\dim(\psi(\varphi(V^{\bot}_i)))=\dim(\varphi(V^{\bot}_i))-\dim(N\cap \varphi(V^{\bot}_i))
\]
This implies 
\[
c_i=\dim(\varphi(V^{\bot}_i)\cap N)=\dim(\varphi(V^{\bot}_i))-\dim(\psi(\varphi(V^{\bot}_i)))
\]
For each $i\in[n]$ let $D_i\in\mathbb{F}^{b\times (b-\dim(V_i))}_q$ denote the matrix such that the linear subspace spanned by its columns is exactly $V^{\bot}_i$ and define the linear map $\phi_i:\F^b_q\to \F^{b-\dim(V_i)}_q$ such that $\phi_i(v)=D^{\top}_iv,\forall v\in\F^b_q$, then $\ker(\phi_i)=V_i$. We compute $c_i$ by definition as follows:
\begin{align*}
\dim(\varphi(V^{\bot}_i))-\dim(\psi(\varphi(V^{\bot}_i)))&=\rank(M^{\top}D_i)-\rank(HM^{\top}D_i)\\
&=\rank(D^{\top}_iM)-\rank(D^{\top}_iMH^{\top})\\
&=\dim(\phi_i(U))-\dim(\phi_i(W))\\
&=(\dim(U)-\dim(U\cap \ker (\phi_i)))-(\dim(W)-\dim(W\cap \ker(\phi_i)))\\
&=(\dim(U)-\dim(U\cap V_i))-(\dim(W)-\dim(W\cap V_i))
\end{align*}
Note that the third line above relies on the facts that the linear subspace spanned by columns of $M$ is $U$ and the linear subspace spanned by columns of $MH^{\top}$, which is the subspace spanned by rows of $HM^{\top}$, is exactly $W$, as discussed in previous paragraphs. It follows that
\begin{align*}
\sum^n_{i=1}c_i&\ge\sum^n_{i=1}(\dim(U)-\dim(U\cap V_i))-(\dim(W)-\dim(W\cap V_i))\\
&=\Phi(\mathcal{V},W,R')-\Phi(\mathcal{V},U,R')-R'n(\dim(W)-\dim(U))\\
&\ge R'n(\dim(U)-\dim(W))\\
&=(R+\mu+1/n)n\dim(N')>(R+\mu)\dim(N')n.
\end{align*}

This contradicts to (\ref{eq:sub-des}), so the assumed pairwise distinct $f_1,\dots,f_b$ do not exist. We conclude that $\mathcal{C}$ does not contain $\mathcal{V}$.
\end{proof}
\section{RLCs Simulate Local Properties Shared by all Optimal Subspace Designable Codes}\label{sec:subspace-to-random}
In this section, we show the easier direction of the local equivalence between subspace designs and random linear codes. Namely, our target is to prove that for any local profile $\mathcal{V}$, if it can be guaranteed to be avoided by all optimal subspace designs, then it can be avoided by random linear code with high probability. As a direct corollary, we show in \cref{cor:sub-lr} that list-recoverability guaranteed by optimal subspace designs also holds for random linear codes. In a follow-up work \cite{BCDZ25lr}, the authors use this reduction to prove the nearly-optimal list-recoverability of random linear codes, as well as random RS codes by a further reduction established in \cite{lms25}.
\paragraph{Technical Overview.} The proof framework is as follows. In \cref{thm:random-subspace-design}, we prove random folded linear codes yield nearly-optimal subspace designs with high probability, so all local properties guaranteed by optimal subspace designs also hold for random folded linear codes with high probability. Then, for any folding parameter $s$ and local profile $\mathcal{V}$, we observe that the values of potential functions (\ref{eq:potential}) satisfy $s\Phi(U,\mathcal{V},R)=\Phi(U,\mathcal{V}^{(s)},R)$, where $\mathcal{V}^{(s)}$ is the $s$-duplicated version of $\mathcal{V}$ which matches the 
$s$-folded codes. Therefore, the rate thresholds satisfy $R_{\mathcal{V}}=R_{\mathcal{V}^{(s)}}$ so local properties of random folded linear codes also hold for the plain random linear codes, which completes the reduction from subspace designs to random linear codes.
\begin{theorem}\label{thm:random-subspace-design}
Fix positive integers $k,n,s,q,d$ where $k\leq sn,d<s$, $\eps\in(0,1)$ and let rate $R=k/(sn)$. A random $\F_q$-linear $C\subseteq (\F^s_q)^n$ is a $(d,(R+\eps)dn)$ subspace designable code with probability at least $1-q^{-\eps sdn+3d^2n+3n}$. 
\end{theorem} 
\begin{proof}
The generator matrix of the random $\F_q$-linear code $C$ can be seen as a random matrix of shape $k\times sn$ whose entries are uniformly and independently sampled from $\F_q$. Let $M_{i}\in\F^{k\times s}_q,i\in[n]$ denote the sub-matrix of this random matrix consisting of its $((i-1)s+1)$-th to $(is)$-th column. We need to show that, with probability at least $1-q^{-\eps sdn+3d^2n+3n}$, for any $d$-dimensional linear subspace of $\F^k_q$ encoded as a matrix $W\in\F^{k\times d}_q$ with rank $d$, there is
\[
f(W):=\sum^n_{i=1}\rank(M^{\top}_iW)\ge (1-R-\eps)dn.
\]
We fix an arbitrary matrix $W\in\F^{k\times d}_q$ with rank $d$. By union bound, since the number of different $W$ is at most $q^{kd}=q^{Rdsn}$, it suffices to prove that with probability at most $q^{-(\eps +R)dsn+3d^2n+3n}$, there is $f(W)<(1-R-\eps)dn$.

For any sequence $A=(a_1,\dots,a_n)\in\{0,1,\dots,d\}^n$, we call $A$ a valid sequence if the sum $S(A):=\sum^n_{i=1}a_i< (1-R-\eps)dn$. Let $E_A$ denote the event that $\rank(M^{\top}_i W)=a_i$ for all $i\in[n]$. We only need to bound the probability that $E_A$ happens for some valid sequence $A$. The number of valid sequences is at most $\binom{(2-R-\eps)n}{n}\leq (e(2-R-\eps))^n\leq (2e)^n\leq q^{3n}$. If we fix some arbitrary valid sequence $A$, by union bound, we only need to show that $E_A$ happens with probability at most $q^{-(\eps +R)dsn+3d^2n}$.

We can directly compute $\Pr[E_A]$ as follows
\[
\Pr[E_A]=\prod^n_{i=1}\Pr[\rank(M^{\top}_iW)=a_i].
\]
Fix any $i\in[n]$, our target is to compute 
$\Pr[\rank(M^{\top}_iW)=a_i]$. since $M_i$ is a uniformly random matrix and $W$ has full column rank.   $\Pr[\rank(M^{\top}_iW)=a_i]$ is the same as the probability that a uniformly random matrix $P\in\F^{s\times d}_q$ has rank $a_i$. It suffices to bound the number of matrices $P\in\F^{s\times d}_q$ with rank $a_i$. If $P$ has rank $a_i$, then there exists a subset $S\subseteq [d]$ of columns of size $|S|=a_i$ such that columns in $S$ are linearly independent, and all the other columns are linear combinations of columns in $S$. Fix such an $S$, the number of matrices $P$ satisfying above conditions is at most $q^{sa_i+a_i(d-a_i)}$. Since there are $\binom{d}{a_i}$ choices of $S$, it follows that
\[
\Pr[\rank(M^{\top}_iW)=a_i]=\Pr[\rank(P)=a_i]\leq\binom{d}{a_i}q^{-sd}q^{sa_i+a_i(d-a_i)}\leq q^{d\log_qd-(s-a_i)(d-a_i)}
\]
Since $\sum^n_{i=1}(d-a_i)\ge (R+\eps)dn$, there is
\[
\Pr[E_A]=\prod^n_{i=1}\Pr[\rank(M^{\top}_iW)=a_i]\leq q^{nd\log_qd-(s-d)(\eps+R)dn}\leq q^{d^2n-(\eps+R)sdn+2d^2n}\leq q^{-(\eps+R)sdn+3d^2n}
\]
\end{proof}
Then we show that if all optimal subspace designable codes does not contain a local profile $\mathcal{V}$, then random linear codes also do not contain $\mathcal{V}$.

\begin{theorem}\label{thm:easy-direction}
Fix $n,b,q,d\ge 1,R\in(0,1)$ and any $b$-local profile $\mathcal{V}=(V_1,\dots,V_n)\in\mathcal{L}(\F^b_q)^n$. If for all large enough $s$, we have that all \textbf{$d$-subspace designable} $s$-folded $\F_q$-linear codes $\mathcal{C}\subseteq (\F^s_q)^n$ with rate $R$ does not contain $\mathcal{V}$, then it follows that $R_{\mathcal{V}}\ge R-(b^2+1)/n$.
\end{theorem}
\begin{proof}
We can choose a large enough $s\ge \Omega(d^3n)$ such that a random $\F_q$-linear code $\mathcal{C}\subseteq (\F^s_q)^n$ is a $(d',Rd'n+1)$ subspace designable code for all $1\leq d'\leq d$ with probability at least $2/3$ by \cref{thm:random-subspace-design}.

let $\mathcal{C}'\subseteq \F^{sn}_q$ denote the ``unfolded'' random $\F_q$-linear code $\mathcal{C}$ with block length $sn$ and define the $s$-duplicated $b$-local profile $\mathcal{V}^{(s)}:=(V_1,\dots,V_1,V_2,\dots,V_2,\dots,V_n,\dots,V_n)\in\mathcal{L}(\F^b_q)^{sn}$ to be the $b$-local profile with length $sn$ derived from $\mathcal{V}$ such that each $V_i$ is repeated $s$ times. From the statement, we know that $\mathcal{C}'$ does not contain $\mathcal{V}^{(s)}$ with probability at least $2/3$.

Then we observe that the threshold rate $R_{\mathcal{V}^{(s)}}=R_{\mathcal{V}}$. The reason is that for any $U\in\mathcal{L}(\F^b_q)$, by the definition of the potential function there is $\Phi(\mathcal{V}^{(s)},U,R)=s\Phi(\mathcal{V},U,R)$, and the rate threshold $R_{\mathcal{V}^{(s)}}$ only depends on signs of $\Phi(\mathcal{V}^{(s)},U,R)-\Phi(\mathcal{V}^{(s)},W,R)$ where $W\subsetneq U, U\in\mathcal{L}(\F^b_q)$. These signs do not change by replacing $\mathcal{V}^{(s)}$ with $\mathcal{V}$.

However, suppose by contrapositive that $R_{\mathcal{V}^{(s)}}=R_{\mathcal{V}}<R-(b^2+1)/n$, we know that the random $\F_q$-linear code $\mathcal{C}'$ with length $sn$ contains $\mathcal{V}^{(s)}$ with probability at least $1-q^{-1}\ge 1/2$ by the first item of \cref{thm:threshold}, which contradicts the fact that  $\mathcal{C}'$ does not contain $\mathcal{V}^{(s)}$ with probability at least $2/3$.
\end{proof}
The second item of \cref{thm:main-equiv} is proved by combining \cref{thm:easy-direction} and the second item of \cref{thm:threshold}.

As a corollary, we conclude that list-recoverability guaranteed by optimal subspace designable codes can also be attained by random linear codes if the alphabet size $q$ is large enough.
\begin{corollary}\label{cor:sub-lr}
Fix constants $\ell,L,d\ge 1,\eps,\rho,R\in(0,1),q\ge(3\ell)^{2(L+1)/\eps}$. Suppose for all large enough $n$, all $d$-subspace designable codes with rate $R$ are $(\rho,\ell,L)$ list-recoverable, then random $\F_q$-linear codes with rate $R-(L^2+2L+2)/n-\eps$ are $(\rho,\ell,L)$ list-recoverable with probability at least $1-o_n(1)$.  
\end{corollary}
\begin{proof}
By \cite[Proposition 2.2]{lms25}, we can construct a series of $(L+1)$-local profiles $\mathcal{V}_1,\dots,\mathcal{V}_m$, where $m=\binom{n}{\rho n}^{L+1}\ell^{(L+1)n}\leq\binom{n}{n/2}^{L+1}\ell^{(L+1)n}\leq (2e)^{(L+1)n/2}\ell^{(L+1)n}\leq (3\ell)^{(L+1)n}$, such that $\F_q$-linear code $\mathcal{C}\subseteq (\F_q)^n$ is $(\rho,\ell,L)$ list-recoverable iff $\mathcal{C}$ does not contain all these $(L+1)$-local profiles. Note that for any $s$-folded $\F_q$-linear code $\mathcal{C}'\subseteq (\F^s_q)^n$, we also know that $\mathcal{C}'$ is $(\rho,\ell,L)$ list-recoverable iff $\mathcal{C}'$ does not contain $\mathcal{V}_1,\dots,\mathcal{V}_m$. (Recall that $\mathcal{C}'$ does not contain some $\mathcal{V}_i$ means the unfolded version of $\mathcal{C}'$ does not contain $\mathcal{V}^{(s)}_i$). Then, by \cref{thm:easy-direction}, we know that for any $\mathcal{V}_i,i\in[m]$, it follows that $R_{\mathcal{V}_i}\ge R-(L^2+2L+2)/n$. Therefore, from \cref{thm:threshold}, for a random $\F_q$-linear code $\mathcal{C}\subseteq (\F_q)^n$, we know that
\[
\forall i\in[m],\quad \Pr[\mathcal{C}\text{ contains }\mathcal{V}_i]\leq q^{-\eps n+(L+1)^2}\leq (3\ell)^{-2(L+1)n+2(L+1)^3/\eps}
\]
By a union bound over all these local profiles, it implies that
\begin{align*}
\Pr[\mathcal{C}\text{ is not }(\rho,\ell,L) \text{ list-recoverable}]&\leq\sum^m_{i=1}\Pr[\mathcal{C}\text{ contains }\mathcal{V}_i]\\
&\leq (3\ell)^{(L+1)n}(3\ell)^{-2(L+1)n+2(L+1)^3/\eps}\\&=(3\ell)^{-(L+1)n+2(L+1)^3/\eps}\leq o_n(1).
\end{align*}
\end{proof}

\section{Applications to Matroid Theory}\label{sec:matroid}

A primary application of our random-to-explicit reduction is contribute new algorithms in matroid theory. We now present a few standard definitions in matroid theory (e.g., \cite{oxley2006matroid}).

A matroid is a pair $\mathcal M = (X, \mathcal I)$, where $X$ is a \emph{ground set} and $\mathcal I$ is a set of subsets of $X$ known as \emph{independent sets} subject to the conditions (1) that $\emptyset \in I$, (2) if $A \in \mathcal I$ and $B \subseteq A$ then $B \in \mathcal I$, and (3) if $A, B \in \mathcal I$, with $|B| > |A|$, then there is $x \in B \setminus A$ with $A \cup \{x\} \in \mathcal I$. Maximal independent sets are called \emph{bases}. Any set $A \subseteq X$ with $A \not \in \mathcal I$ is \emph{dependent}. In particular, a dependent set $A$ is a \emph{circuit} if every proper subset of $A$ is independent (that is $A$ is a minimal dependent set). We say that any $A \subseteq X$ containing a basis is \emph{spanning}.

Every matroid $\mathcal M$ has an associated \emph{rank function} $r_{\mathcal M}$, where for any $A \subseteq X$, $r_{\mathcal M}(A)$ is the size of the largest independent set $B \in \mathcal I$ which is a subset of $A$. In particular, $A \subseteq X$ is independent if and only if $r_{\mathcal M}(A) = |A|$. We define $r_{\mathcal M}(X)$ (i.e., the size of every basis) to be the \emph{rank} of $\mathcal M$. Given two matroids $\cM_1 := (X, \cI_1)$ and $\cM_2 = (X, \cI_2)$ on the same ground set, we say that $\cM_1 \preceq \cM_2$ if $\cI_1 \subseteq \cI_2$. In other words, $r_{\mathcal M_1}(A) \le r_{\mathcal M_2}(A)$ for all $A \subseteq X$. Given matroids $\mathfrak M$ on the same ground set, we say that $\cM \in \mathfrak M$ is \emph{maximal} if $\cN \not\succ \cM$ for all $\cN \in \mathfrak M$. 

By interpreting certain matroid properties as local profiles, we apply our random-to-explicit reduction to obtain novel \emph{conditional} deterministic algorithms for computing the rank functions of various matroids. In this paper, we focus on the following three families of matroids.

\paragraph{Maximally Recoverable (MR) Tensor Code Matroid.} Given a field $\F$, an $[n,k]$-code is a $k$-dimensional subspace $C$ of $\F^n$. Motivated by applications in distributed storage, Gopalan et al.~\cite{Gopalan2016} studied the following code architecture. Pick parameters, $m,n,a,b \in \N$ and consider a $[m,m-a]$-column code $C_{\col}$ and a $[n,n-b]$ row code $C_{\row}$. The tensor product of these codes $C_{\col} \otimes C_{\row} := \operatorname{span}\{c_{\col} \otimes c_{\row} \mid c_{\col} \in C_{\col}, c_{\row} \in C_{\row}\} \subseteq \F^{m \times n}$ is the set of all $m \times n$ matrices with entries from $\F$ where every row is a codeword in $C_{\row}$ and every column is a codeword in $C_{\col}$.

The tensor code $C_{\col} \otimes C_{\row}$ has an associated correctability matroid on the ground set $[m] \times [n]$, where a \emph{pattern} $E \subseteq [m] \times [n]$ is independent in the matroid if and only if one can recover from erasing the symbols associated with $E$. In other words, the punctured\footnote{Given a code $C \subseteq \F^n$ and a set $A \subseteq [n]$, we let $C|_{A} := \{c|_{A} : c \in C\}\subseteq \F^A$ denote a puncturing of $C$.} code $(C_{\col} \otimes C_{\row})|_{\overline{E}}$ has dimension $(m-a)(n-b)$. We call this matroid $\cM(C_{\col}, C_{\row})$.

For a fixed choice of $m,n,a,b$, there are many possible matroids $\cM(C_{\col}, C_{\row})$. However, there is a provably unique maximal matroid of this form, which corresponds to the setting in which $\F$ is a characteristic zero field (e.g., $\Q, \R,$ or $\C$) and $C_{\col}$ and $C_{\row}$ are chosen uniformly from a continuous distribution, which we call \emph{generic}\footnote{More precisely, ``generic'' refers to properties which are Zariski dense (in the sense of algebraic geometry). For our purposes, a ``generic'' code is a random linear code for which the relevant matroid properties hold with probability $1$.} codes. In this paper, we denote this unique maximal matroid by $\MR(m,n,a,b)$. Since $C_{\col} \otimes C_{\row}$ has dimension $(m-a)(n-b)$, the rank of this matroid is $mn - (m-a)(n-b) = bm+an-ab$.

\paragraph{Bipartitie Rigidity (Birigidity) Matroid.} Kalai et al.~\cite{kalai2015Bipartite} posed the following geometric question, which is a  variant of the classical graph rigidity problem studied since the 19th century~\cite{Maxwell}. Consider a bipartite graph $G = (V, E)$ where $V$ is the disjoint union of $[m]$ and $[n]$ and $E \subseteq [m] \times [n]$. Consider two embedding maps $\alpha : [m] \to \R^{b} \times \{0^a\}$ and $\beta : [n] \to \{0^b\} \times \R^a$. Kalai et al.~\cite{kalai2015Bipartite} say this embedding is \emph{rigid} if no infinitesimal change to $\alpha$ and $\beta$ can preserve the ($\ell_2$) length of every edge $E$. It turns out such rigid graphs correspond to the spanning sets of a bipartite rigidity matroid with parameters $(m,n,a,b)$. Like in the case of $\MR$, if $\alpha$ and $\beta$ are chosen generically, there is a unique maximal \emph{generic birigidity matroid}, which we denote by $\BR(m,n,a,b)$ in this paper. See \cite{kalai2015Bipartite,jackson2024maximal,brakensiek2024Rigidity,cruickshank2025rigidity} for a more precise description of this matroid.

Recently Brakensiek et. al.~\cite{brakensiek2024Rigidity} proved all the $\MR$ and $\BR$ matroids are equal, so any properties of one such matroid apply to the other.

\begin{theorem}[Theorem~1.1(3) of \cite{brakensiek2024Rigidity}, restated]
For all $m, n, a, b \in \mathbb N$ with $m \ge a$ and $n \ge b$, we have that $\MR(m,n,a,b) = \BR(m,n,a,b)$.
\end{theorem}

\paragraph{Matrix Completion.} In the case where $a=b$, $\BR(m,n,a,a)$ has an alternative interpretation as the \emph{matrix completion problem} (see \cite{singer2010uniqueness,cruickshank2025rigidity}). In particular, for $E \subseteq [m] \times [n]$, imagine that the entries indexed by $E$ are revealed in a (say, complex-valued) $m \times n$ matrix. For a generic (aka random) choice of these revealed entries, does there exist a rank $a$ matrix $M \in \C^{m \times n}$ which corresponds to the revealed entries? Such questions are closely related to the \emph{Netflix problem} in machine learning (e.g.,~\cite{nguyen2019low,dzhenzher2025low}).

For all three of these families of matroids, it is a substantial open question to give a \emph{deterministic} polynomial time algorithm for detecting whether a pattern $E \subseteq [m] \times [n]$ is independent in these matroids (e.g., \cite{Gopalan2016,holzbaur2021correctable,bgm2021mds,jackson2024maximal,brakensiek2024Rigidity,cruickshank2025rigidity}). Such an algorithm is only known when $a=1$ \cite{whiteley1989matroid,kalai2015Bipartite,Gopalan2016,bgm2021mds,brakensiek2023generic} or $m-a \le 3$ \cite{brakensiek2024Rigidity}. We emphasize that all three of these matroids have efficient \emph{randomized} algorithms by invoking the Schwarz-Zippel lemma. In particular, the case $a=b=2$ (aka rank $2$ matrix completion) is unresolved~\cite{bernstein2017completion,brakensiek2024Rigidity,cruickshank2025rigidity}. However, barring the resolution of $\mathsf{RP} = \mathsf{P}$, this approach will not yield an efficient deterministic algorithm.

In an alternative direction, Jackson and Tanigawa~\cite{jackson2024maximal} recently proved that the computational complexity of these rank functions lies in $\mathsf{coNP}$ if one makes a \emph{conditional} assumption on the structure of these rigidity matroids (see \cite[Conjecture 6.4]{jackson2024maximal}). Improving on this result, we conditionally show (under a weaker assumption) that detecting independence lies in $\mathsf{P}$. More precisely, our matroid-theoretic assumption is a positive answer to a open question of Mason~\cite{mason1981glueing}, which is formalized in Conjecture~\ref{conj:birigid}.

\begin{theorem}\label{thm:MR-algo}
Assuming \Cref{conj:birigid}, Algorithm~\ref{algo:MR-rank} computes whether $E \subseteq [m] \times [n]$ is independent in $\MR(m,n,a,b) = \BR(m,n,a,b)$ in deterministic $\poly(m,n)$ time.
\end{theorem}

By a standard oracle reduction (e.g., \cite{oxley2006matroid}), this also implies a deterministic $\poly(m,n)$-time algorithm for the rank function of $\MR(m,n,a,b)$.

\subsubsection*{Proof Overview}

The proof of \Cref{thm:MR-algo} has a number of components which we outline as follows. Building on \Cref{sec:random-to-explicit}, our high-level goal is to show that the tensor product of two explicit folded Reed-Solomon codes ``simulates'' the $\MR(m,n,a,b)$ matroid. However, the accuracy of this simulation is contingent on the validity of \Cref{conj:birigid}.

\paragraph{\Cref{subsec:abstract}.} First, we build up the necessary background needed to formally discuss \Cref{conj:birigid}. In particular, we define the notion of an \emph{abstract birigidity matroid} \cite{jackson2024maximal,cruickshank2025rigidity} which captures matroids which ``look like'' birigidity matroids.  \Cref{conj:birigid} formalizes the intuition of matroid theorists that $\MR(m,n,a,b)$ is the (unique) maximal abstract birigidity matroid. We remark that  \Cref{conj:birigid} is equivalent to a positive answer to a much older question due to Mason~\cite{mason1981glueing}.

\paragraph{\Cref{subsec:potential-matroid}.} Inspired by the potential functions of \cite{lms25}, we define a potential function that is (unconditionally) at least the rank function of $\MR(m,n,a,b)$. We show this potential function satisfies the axioms of a matroid and that this matroid is an $(m,n,a,b)$ abstract birigidity matroid. As such, by Conjecture~\ref{conj:birigid}, we have that this potential function \emph{is} the rank function of $\MR(m,n,a,b)$.

\paragraph{\Cref{subsec:pass-to-FRS}.} Even with these facts in hand, directly computing this potential function is still difficult, as it corresponds to the behavior of the tensor product of two random linear codes. By adapting the framework of \Cref{sec:random-to-explicit}, we show that this potential function can be approximated by the tensor product of two explicit folded\footnote{Although we are using the theory of folded Reed--Solomon codes, in the actual analysis we ``unfold'' the Reed--Solomon codes and adjust the erasure patterns accordingly through what we call \emph{scaling}.} Reed-Solomon codes. To do this, we actually need a two-phase approach. First, we show that the tensor product of two random linear codes is closely approximated by the tensor product of one random linear code with a folded Reed-Solomon code. Then, we swap the roles of the two codes being tensored and construct a new potential function which allows us to replace the remaining random linear code with another folded Reed Solomon code. Unlike most applications of folded Reed Solomon codes, the folding parameters need to be of size $\poly(m,n)$ so that the folded Reed-Solomon sufficiently approximate the worst-case rank loss by the FRS codes.

\paragraph{\Cref{subsec:proof-thm}.} Finally, we present Algorithm~\ref{algo:MR-rank}, which constructs two folded Reed-Solomon codes, takes their tensor product, and performs a (linear algebra) rank computation to obtain the rank function of $\MR(m,n,a,b)$. The analysis is a straightforward combination of the tools built in the previous sections.

\subsection{Abstract Bipartite Rigidity Matroids}\label{subsec:abstract}

For the remainder of this section, we assume $\F$ is a field with characteristic $0$ (e.g., $\Q, \R, \C$). To properly discuss \Cref{conj:birigid}, we need to introduce the notion of an \emph{abstract birigidity matroid} (see \cite{jackson2024maximal,cruickshank2025rigidity}). There are a few equivalent definitions (see Theorem 5.29 of \cite{cruickshank2025rigidity}), but the following is the most convenient for our use case.

\begin{definition}[e.g., \cite{jackson2024maximal,cruickshank2025rigidity}]
For parameters $(m,n,a,b)$, we say that a matroid $\cM$ on round set $[m] \times [n]$ is an $(m,n,a,b)$-abstract birigidity matroid if it meets the following conditions.
\begin{itemize}
\item The rank of $\cM$ is $bm+an-ab$.
\item For every $A \subseteq [m]$ of size $a+1$ and every $B \subseteq [n]$ of size $b+1$, we have $A \times B$ is a circuit.
\end{itemize}
\end{definition}

To gain some intuition, we first prove that ordinary tensor codes are an example of abstract birigidity matroids if we make the assumption that $C_{\col}$ and $C_{\row}$ are \emph{MDS} codes.  In particular, an $[n,k]$-code $C \subseteq \F^n$ is MDS it attains the Singleton bound: every nonzero codeword has Hamming weight at least $n-k+1$. In matroid theory, such MDS codes correspond to \emph{uniform} matroids, where every subset of $[n]$ of size $n-k$ is a basis.

\begin{proposition}\label{prop:tensor-abstract}
For any $[m, m-a]$-MDS code $C_{\col}$ and $[n,n-b]$-MDS code $C_{\row}$, we have that the correctability matroid for $C_{\col} \otimes C_{\row}$ is an $(m,n,a,b)$ abstract birigidity matroid. Thus, $\MR(m,n,a,b)$ is an $(m,n,a,b)$ abstract birigidity matroid.
\end{proposition}
\begin{proof}
As previously discussed, the rank of the correctability matroid is $mn - (m-a)(n-b) = bm+an-ab$, as desired. It suffices to check for every $A \subseteq [m]$ of size $a+1$ and $B \subseteq [n]$ of size $b+1$ that $A \times B$ is a circuit. 

First, too see that $A \times B$ is dependent in $C_{\col} \otimes C_{\row}$, note that since $C_{\col}$ is $m-a$-dimensional, there exists at least one nonzero $u \in C_{\col}$ supported on $A$. Likewise, there is at least one $v \in C_{\row}$ supported on $B$. The codeword $u \otimes v$ is supported entirely on $A \times B$. Thus, erasing $A \times B$ makes $u \otimes v$ unrecoverable. Therefore, $A \times B$ is dependent.

However, consider any $(a, b) \in A \times B$, we claim that $A \times B \setminus \{(a,b)\}$ is independent. The $b$th column has exactly $a$ erasures. Thus, by using the column parity checks of $C_{\col}$, we can recover the entire column $[m] \times \{b\}$. This leaves our erasure pattern as $A \times (B \setminus \{b\})$. Every row of this pattern has at most $b$ erasures, so we can use the parity checks of $C_{\row}$ to recover the rest of the codeword. Thus, $A \times B \setminus \{(a,b)\}$ is independent.
\end{proof}

\begin{conjecture}[Conjecture 5.30~\cite{cruickshank2025rigidity}, restated]\label{conj:birigid}
$\BR(m,n,a,b)=\MR(m,n,a,b)$ is the maximal $(m,n,a,b)$ abstract birigidity matroid.
\end{conjecture}

This conjecture is equivalent to an affirmative answer to an open question raised by Mason~\cite{mason1981glueing} over 40 years ago on the most general possible tensor product of two uniform matroids. We also note that the conditional assumption made by Jackson and Tanigawa~\cite{jackson2024maximal} is a stronger assumption than \Cref{conj:birigid} that particular rank estimation techniques in the generic bipartite rigidity matroid are perfectly accurate~\cite{Larson}.

\begin{remark}\label{rem:conj-pro-con}
We now discuss the current evidence in the literature for and against \Cref{conj:birigid}. Much of this evidence is presented in Section 5 of the recent rigidity survey~\cite{cruickshank2025rigidity}.

Currently, the strongest evidence for \Cref{conj:birigid} being true is that it has been proved in the special cases when $a=1$ \cite{whiteley1989matroid,kalai2015Bipartite,Gopalan2016} or $m-a \le 3$ \cite{brakensiek2024Rigidity} (with other parameters arbitrary). Of note, the truth of \Cref{conj:birigid} when $a=1$ led \cite{brakensiek2023generic} to connect the structure of higher order MDS codes to the recent algorithms for computing non-commutative ranks \cite{GGOW20}. We note that the resolution of these cases $a=1$ and $m-a\le 3$ imply Algorithm~\ref{algo:MR-rank} is unconditionally correct in these settings, although efficient deterministic algorithms were already known in such cases~\cite{bgm2021mds,brakensiek2024Rigidity}.

Conversely, the strongest evidence against \Cref{conj:birigid} is that the natural analogue of the conjecture for (traditional) non-bipartite rigidity (where the entire graph is embedded into the same space) is false when the ambient dimension is 4 or greater. In fact, such counterexamples have been known for over 30 years--see the discussion surrounding Conjecture 5.7.1 in \cite{graver1993combinatorial} as well as the discussion in Section 5.3.2 of \cite{cruickshank2025rigidity}. However, no analogue of these constructions has been found in the non-bipartite setting which disproves \Cref{conj:birigid}.

To mitigate such concerns, we do not need the full power of \Cref{conj:birigid} to prove Theorem~\ref{thm:MR-algo}. In particular, the maximality of $\MR(m,n,a,b)$ is not strictly required; rather, we merely need that $\MR(m,n,a,b)$ is equal to the corresponding ``potential matroid'' which we construct in \Cref{subsec:potential-matroid}. Thus, if one is pessimistic about the validity of \Cref{conj:birigid}, one could use the novel matroid construction in \Cref{subsec:potential-matroid} as a new plan of attack toward generating counterexamples toward \Cref{conj:birigid} and negatively answering the question of Mason~\cite{mason1981glueing}.

We also briefly touch on the coding theory implication if the fragment of \Cref{conj:birigid} we use is false. In short, it would imply that the tensor product of two (random) folded linear codes has \emph{strictly} better\footnote{In fact, the folded RS codes will have a slightly smaller rate. However, it is a general phenomenon in coding theory that explicitly simulating random linear codes requires a slack from the threshold, so an arbitrarily small rate decrease as in our case should be in expectation. See Section III of \cite{BDGZ24} for a theory of \emph{approximate} MR tensor codes with a rate slack.} maximally recoverability guarantees than the tensor product of two (random) linear codes. In other words, for practical applications (see discussion in \cite{Gopalan2016}), it may be preferable to consider a folded linear architecture when designing tensor codes. For example, in \cite[Section 6]{ccs25}, the authors use the tensor product of two explicit folded linear codes to construct graph codes.
\end{remark}

\subsection{The Potential Matroid}\label{subsec:potential-matroid}

Fix parameters $m,n,a,b\in \mathbb N$. Fix an $[m, m-a]$ code $C_{\col} \in \F^m$. Given a set $A \subseteq [m]$, we let $\F^A \subseteq \F^m$ denote the vector space spanned by the indicator vectors indexed by $A$.  Inspired by \cite{lms25}, we now define the following potential function $\Phi$ on subsets $E \subseteq [m] \times [n]$ as follows.
\begin{align}
\Phi(E) := \max_{U \subseteq C_{\col}}\left[-b \dim(U) + \sum_{j=1}^n \dim(U \cap \F^{E_j})\right],\label{eq:Phi}
\end{align}
where for all $j \in [n]$, $E_j := \{i \in [m] \mid (i,j) \in E\}.$ Further define
\begin{align}
r(E) := |E| - \Phi(E).\label{eq:r}
\end{align}
Our first goal is to prove that $r$ is the rank function of a matroid. Toward this, we need the following subspace inequality.  

\begin{claim}\label{claim:supermodular}
For any subspaces $X_1, X_2, Y_1, Y_2 \subseteq \F_q^m$, we have that
\[
    \dim(X_1 \cap Y_1) + \dim(X_2 \cap Y_2) \le \dim ((X_1 \cap X_2) \cap (Y_1 \cap Y_2)) + \dim((X_1+X_2) \cap (Y_1+Y_2)).
\]
\end{claim}
\begin{proof}
Observe that
\begin{align*}
\dim(X_1 \cap Y_1) + \dim(X_2 \cap Y_2) - \dim ((X_1 \cap X_2) \cap (Y_1 \cap Y_2)) &= \dim((X_1 \cap Y_1) + (X_2 \cap Y_2))\\
&\le \dim((X_1 + X_2) \cap (Y_1 + Y_2)),
\end{align*}
where the last inequality follows from the fact that $X_1 \cap Y_1 \subseteq (X_1 + X_2) \cap (Y_1 + Y_2)$ and $X_2 \cap Y_2 \subseteq (X_1 + X_2) \cap (Y_1 + Y_2).$
\end{proof}

We now prove that $r$ is the rank function of a matroid.

\begin{lemma}\label{lem:potential-matroid}
There exists a matroid $\mathcal M(C_{\col}, n, b)$ over $[m] \times [n]$ with rank function $r$.
\end{lemma}
\begin{proof}

To show that $r$ is the rank function of a matroid, we have to prove the following properties (see Corollary 1.3.4 from \cite{oxley2006matroid})
\begin{enumerate}
\item $r(E) \le |E|$ for all $E \subseteq [m] \times [n]$.\label{item:2}
\item $r(E) \ge 0$ for all $E \subseteq [m] \times [n]$.\label{item:1}
\item For all $E \subseteq F$, we have that $r(E) \le r(F)$.\label{item:4}
\item $r$ is submodular: for all $E, F \subseteq [m] \times [n]$, we have that $r(E \cap F) + r(E \cup F) \le r(E) + r(F)$.\label{item:3}
\end{enumerate}
\paragraph{\Cref{item:2}.} First, by picking $U = 0$, we can see that $\Phi(E) \ge 0$. Thus, \Cref{item:2} immediately holds.

\paragraph{\Cref{item:1}.} Further, for any $U \subseteq C_{\col}$, we have that
\[
- b \dim(U) + \sum_{j=1}^n \dim(U \cap \F^{E_j}) \le \sum_{j=1}^n |E_j| = |E|.
\]
So, $\Phi(E) \le |E|$, or $r(E) \ge 0$, proving \Cref{item:1}.

\paragraph{\Cref{item:4}.} For $E \subseteq F$,
to show that $r(E) \le r(F)$, note that for any $U \subseteq C_{\col}$, we have that
\begin{align*}
- b \dim(U) + \sum_{j=1}^n \dim(U \cap \F^{E_j}) &\ge - b \dim(U) + \sum_{j=1}^n \left[\dim(U \cap \F^{F_j}) - |F_j \setminus E_j|\right]\\
&=- |F \setminus E| -b \dim(U) + \sum_{j=1}^n \dim(U \cap \F^{F_j}).
\end{align*}
This implies that $\Phi(E) \ge \Phi(F) - |F \setminus E|$, which implies that $r(E) \le r(F)$, proving \Cref{item:4}.

\paragraph{\Cref{item:3}.} Pick $U_E, U_F \subseteq C_{\col}$ such that

\begin{align*}
\Phi(E) &= - b \dim(U_E) + \sum_{j=1}^n \dim(U_E \cap \F^{E_j})\\
\Phi(F) &= - b \dim(U_F) + \sum_{j=1}^n \dim(U_F \cap \F^{F_j})
\end{align*}
Further observe that
\begin{align*}
\Phi(E \cap F) &\ge - b \dim(U_E \cap U_F) + \sum_{j=1}^n \dim((U_E \cap U_F) \cap \F^{E_j \cap F_j})\\
\Phi(E \cup F) &\ge - b \dim(U_E + U_F) + \sum_{j=1}^n \dim((U_E + U_F) \cap \F^{E_j \cup F_j}).
\end{align*}
Note that $\F^{E_j\cap F_j} = \F^{E_j} \cap \F^{F_j}$ and $\F^{E_j \cup F_j} = \F^{E_j} + \F^{F_j}$. Using \cref{claim:supermodular}, we have that
\begin{align*}
\Phi(E \cap F) + \Phi(E \cup F)
&\ge - b \dim(U_E \cap U_F) + \sum_{j=1}^n \dim((U_E \cap U_F) \cap (\F^{E_j} \cap \F^{F_j}))\\
&\ \ \ \ \ \ \ \ - b \dim(U_E + U_F) + \sum_{j=1}^n \dim((U_E + U_F) \cap (\F^{E_j} + \F^{F_j}))\\
&\ge -b \dim(U_E) - b\dim(U_F) + \sum_{j=1}^n \left[\dim(U_{E} \cap \F^{E_j}) + \dim(U_F \cap \F^{F_j})\right]\\
&= \Phi(E) + \Phi(F).
\end{align*}
Since $|E| + |F| = |E \cap F| + |E \cup F|$, we have that $r(E \cap F) + r(E \cup F) \le r(E) + r(F)$, proving \Cref{item:3}.

Thus, $r$ is indeed the rank function of a matroid.
\end{proof}

We remark that $\Phi(E) = |E| - r(E)$ is called the \emph{corank} function of $\cM(C_{\col}, n, b)$. We next show that $\cM(C_{\col}, n, b)$ has more independent sets than $\cM(C_{\col}, C_{\row})$.

\begin{lemma}\label{lemma:potential-monotone}
Let $C_{\row} \subseteq \F^n$ be an arbitrary $[n, n-b]$ code. If $E$ is a correctable erasure pattern for $C_{\col} \otimes C_{\row}$, then $E$ is an independent set of $\cM(C_{\col}, n, b)$. 
\end{lemma}

\begin{proof}
We prove the contrapositive: assume that $E$ is a dependent set of $\cM(C_{\col}, n, b)$, we seek to show that the erasure pattern $E$ is not correctable in $C_{\col} \otimes C_{\row}$. Thus, $\Phi(E) \ge 1$. This means there exists $U \subseteq C_{\col}$ such that 
\[
  -b \dim(U) + \sum_{j=1}^n \dim(U \cap \F^{E_j}) \ge 1.
\]
We now invoke the a dimension-counting argument. Let $\F^m \otimes C_{\row}$ be the set of $m \times n$ matrices each of whose $m$ rows lies in $C_{\row}$. Let $\bigoplus_{j=1}^n (U \cap \F^{E_j})$ be the set of $m \times n$ matrices, such that for all $j \in [n]$, the $j$th row lies in $U$ and is supported on $E_j$. We see that
\begin{align*}
  \dim&\left[\F^m \otimes C_{\row} \cap \bigoplus_{j=1}^n (U \cap \F^{E_j})\right]\\
&= \dim\left[U \otimes C_{\row} \cap \bigoplus_{j=1}^n (U \cap \F^{E_j})\right]\\
&= \dim(U \otimes C_{\row}) + \sum_{j=1}^n \dim(U \cap \F^{E_j}) - \dim\left[U \otimes C_{\row} + \bigoplus_{j=1}^n (U \cap \F^{E_j})\right]\\
&\ge (n-b)\dim(U) + (b\dim(U) + 1) - \dim(U \otimes \F^n)\\
&= 1,
\end{align*}
Thus, there exists a nonzero codeword $c \in \F^m \otimes C_{\row} \cap \bigoplus_{j=1}^n (U \cap \F^{E_j}) \subseteq C_{\col} \otimes C_{\row}$. In other words, $C_{\col} \otimes C_{\row}$ has a nonzero codeword supported on $E$. Therefore, $E$ is not a correctable erasure pattern in $C_{\col} \otimes C_{\row}$, as desired.
\end{proof}

Next, we prove that $\cM(C_{\col}, n, b)$ is an abstract rigidity matroid when $C_{\col}$ is MDS.

\begin{lemma}\label{lem:potential-abstract}
If $C_{\col}$ is an MDS $[m, m-b]$ code, then $\cM(C_{\col}, n, b)$ is an $(m,n,a,b)$ abstract rigidity matroid.
\end{lemma}

\begin{proof}
First, we need to show that the rank of $\cM(C_{\col}, n, b)$ is $bm+an-ab$. To see why, observe that for all $U \subseteq C_{\col}$, we have that
\[
  -b\dim(U) + \sum_{j=1}^n \dim(U \cap \F^{[m]}) = (n-b)\dim(U).
\]
Thus, $\Phi([m] \times [n]) = (n-b)\dim(C_{\col}) = (n-b)(m-a)$. Thus, $r([m] \times [n]) = mn - (m-a)(n-b) = bm+an-ab$, as desired.

Second, we seek to show that for any $A \subseteq [m]$ of size $a+1$ and $B \subseteq [n]$ of size $b+1$ that $A \times B$ is a circuit of $\cM(C_{\col}, n, b)$. First, to see why $A \times B$ is dependent, let $U_A \subseteq C_{\col}$ be the set of codewords supported on $A$. Since $C_{\col}$ is an $[m, m-a]$ MDS code and $|A| = a+1$, we have that $\dim(U_A) = 1$. Thus,
\[
  \Phi(A \times B) \ge -b \dim(U_A) + \sum_{j \in B} \dim(U_A \cap \F^{A}) = -b + (b+1) = 1.
\]
Thus, $r(A \times B) < |A \times B|$, so $A \times B$ is dependent.

Pick any $a \in A$ and $b \in B$. By \Cref{prop:tensor-abstract}, we have that $A \times B \setminus \{(a,b)\}$ is an independent set of $C_{\col} \times C_{\row}$ for any $[n, n-b]$ MDS code $C_{\row} \subseteq \F^n$. Thus, by \Cref{lemma:potential-monotone}, we have that $A \times B \setminus \{(a,b)\}$ is an independent set of $\cM(C_{\col}, n, b)$. 
\end{proof}

We conclude this section with the following observation. If $C_{\col} \subseteq \F^m$ is a generic $[m, m-a]$ code and $C_{\row} \subseteq \F^n$ is a generic $[n, n-b]$ code, then both codes are MDS (e.g., \cite{bgm2021mds}). Further, we have by \Cref{lemma:potential-monotone} that $\cM(C_{\col}, n, b) \succeq \cM(C_{\col}, C_{\row}) = \MR(m, n, a, b)$. Since $\cM(C_{\col}, n, b)$ is an $(m,n,a,b)$ abstract tensor matroid by \Cref{lem:potential-abstract}, we have the following corollary of \Cref{conj:birigid}.

\begin{corollary}\label{cor:potential-MR}
Assume \Cref{conj:birigid}. Let $C_{\col} \subseteq \F^m$ be a generic $[m, m-a]$ code, then $\cM(C_{\col}, n, b) = \MR(m, n, a, b)$.
\end{corollary}

This corollary will be of use in the next section, where we seek to replace $C_{\row}$ (and later $C_{\col}$) with a folded Reed-Solomon code.

\subsection{Replacing Generic Codes with FRS} \label{subsec:pass-to-FRS}

Building on the techniques of \Cref{sec:random-to-explicit}, we show how to replace the generic codes $C_{\col}$ and $C_{\row}$ with suitable (folded) Reed-Solomon codes. More precisely, fix parameters $m, n, a, b \in \N$. Fix parameters $t := 2m^2n$ and $d := mn$ and construct the following datum:
\begin{itemize}
\item Let $b' := bt + d$.
\item Pick $\beta_1, \hdots, \beta_n \in \F$ and $\gamma \in \F$ such that the set $Z := \{\alpha_i \gamma^j \mid i \in [n], j \in \{0,1 , \hdots, t-1\}\}$ has exactly $tn$ elements.
\item Let $C^{\RS}_{\row} \subseteq \F^{tn}$ be the $[tn, tn - b']$ Reed-Solomon code with evaluation points from $Z$. That is, $C^{\RS}_{\row} = \RS^{\F}_{tn, tn-b'}(\beta_1, \gamma\beta_1, \hdots, \gamma^{t-1}\beta_1, \hdots, \beta_n, \gamma\beta_n, \hdots, \gamma^{t-1}\beta_n).$
\end{itemize}
One acceptable choice of parameters (as done in \Cref{algo:MR-rank}) is to set $\gamma = 2$ and let $\beta_1, \hdots, \beta_n$ be consecutive odd integers.

We seek to compare the matroid $\cM(C_{\col}, C_{\row})$ with $\cM(C_{\col}, C^{\RS}_{\row})$. However, the matroids are on different ground sets. To correct for this, we need to ``scale'' an erasure pattern $E \subseteq [m] \times [n]$ as follows.

\begin{definition}[Pattern scaling]
Given $E \subseteq [m] \times [n]$ and parameter $s,t \in \mathbb N$, we define the $(s,t)$-scaling of $E$ to be the pattern $E^{s,t} \subseteq [sm] \times [tn]$ such that
\begin{itemize}
\item if $(i,j) \in E$ then $(s(i-1) + x, t(j-1) + y) \in E^{s,t}$ for all $(x,y) \in [s] \times [t]$, and
\item if $(i,j) \not\in E$ then $(s(i-1) + x, t(j-1) + y) \not\in E^{s,t}$ for all $(x,y) \in [s] \times [t]$.
\end{itemize}
\end{definition}

The main goal of this section is to prove the following key ``scaling'' lemma.

\begin{lemma}\label{lem:replace-FRS}
For all $E \subseteq [m] \times [n]$ we have that $E$ is an independent set of $\cM(C_{\col}, n, b)$ if and only $E^{1,t}$ is an independent set of $\cM(C_{\col}, C^{\RS}_{\row}).$
\end{lemma}

The proof of Lemma~\ref{lem:replace-FRS} requires a careful adaptation of the techniques of \Cref{sec:random-to-explicit}. In some sense, the adapted argument in this section is simpler, as we do not require the distinctness condition needed in \Cref{sec:random-to-explicit}.

\subsubsection{Dependent Sets Are Preserved}

As a first step toward proving \Cref{lem:replace-FRS}, we show that if $E$ is a dependent set of $\cM(C_{\col}, n, b)$, then $E^{1,t}$ is a dependent set of $\cM(C_{\col}, C^{\RS}_{\row}).$ The proof is an application of \Cref{lemma:potential-monotone}.

\begin{proposition}\label{prop:RS-dependent}
If $E$ is a dependent set of $\cM(C_{\col}, n, b)$ then  $E^{1,t}$ is a dependent set of $\cM(C_{\col}, nt, b')$ and thus a dependent set of $\cM(C_{\col}, C^{\RS}_{\row}).$
\end{proposition}

\begin{proof}
Since $E$ is a dependent set of $\cM(C_{\col}, n, b)$, we have that $\Phi(E) \ge 1$. In other words, for all $U \subseteq C_{\col}$, we have that
\[
  \sum_{j=1}^n \dim(U \cap \F^{E_j}) \ge b\dim(U) + 1 \ge \left(b + \frac{1}{m}\right)\dim(U).
\]
Let $\Phi'$ be the corresponding corank function of the matroid $\cM(C_{\col}, tn, b')$. Note then that
\begin{align*}
\Phi'(E^{1,t}) &= \max_{U \subseteq C_{\col}} \left[-b' \dim(U) + \sum_{j=1}^{tn} \dim(U \cap \F^{E^{1,t}_j})\right]\\
&= t\max_{U \subseteq C_{\col}} \left[-\frac{b'}{t}\dim(U) + \sum_{j=1}^n \dim(U \cap \F^{E_j})\right]\\
&\ge t\max_{U \subseteq C_{\col}} \left[-\frac{b'}{t}\dim(U) + \left(b + \frac{1}{m}\right)\dim(U) \right]\\
&= tm\left(b + \frac{1}{m}- \frac{b'}{t}\right) > 0,
\end{align*}
where we use the fact that $b' = bt + d < bt + \frac{t}{m}$. Thus, $E^{1,t}$ is dependent in $\cM(C_{\col}, nt, b')$.  By  \Cref{lemma:potential-monotone}, we have that $E^{1,t}$ is dependent in $\cM(C_{\col}, C^{\RS}_{\row})$, as desired.
\end{proof}

\subsubsection{Independent Sets Are Preserved}

Next, we show the converse of \Cref{prop:RS-dependent}, which is sufficient to prove \Cref{lem:replace-FRS}.

\begin{proposition}\label{prop:RS-independent}
If $E^{1,t}$ is a dependent set of $\cM(C_{\col}, C^{\RS}_{\row})$ then $E$ is a dependent set of $\cM(C_{\col}, n, b)$.
\end{proposition}

Similar to argument in \Cref{sec:random-to-explicit}, the proof of \Cref{prop:RS-independent} crucially uses the subspace design properties underlying $C^{\RS}_{\row}$. We use the following variant of the subspace design conditions proved in \Cref{sec:subspace-design}.

\begin{proposition}\label{prop:RS-subspace-design}
For $j \in [n]$, let $I_j := \{(j-1)t+1, (j-1)t+2, \hdots, jt\}$. For any linear subcode $V \subseteq C^{\RS}_{\row}$ for which $\dim(V) \le t$, we have that
\[
  \sum_{j=1}^n \dim(V|_{I_j}) \ge n\dim(V) - \frac{\dim(V) (tn-b' - \dim(V))}{t - \dim(V) + 1}.
\]
\end{proposition}

We defer the proof of \Cref{prop:RS-subspace-design} to \Cref{subsec:deferred-proof}.

\begin{proof}[Proof of \Cref{prop:RS-independent}.]
Assume that $E^{1,t}$ is a dependent set of $\cM(C_{\col}, C^{\RS}_{\row})$. Thus, there is a nonzero matrix $M \in C_{\col} \times C^{\RS}_{\row}$ which is supported on $E^{1,t}$. Let $U_M := \operatorname{colspan}(M) \subseteq C_{\col}$. We seek to show that
\begin{align}
 - b\dim(U_M) + \sum_{j=1}^n \dim(U_M \cap \F^{E_j}) > 0,\label{eq:potential-goal}
\end{align}
as this would prove that $E$ is a dependent set of $\cM(C_{\col}, n, b)$. Let $V_M := \operatorname{rowspan}(M) \subseteq C^{\RS}_{\row}$, and note that $\dim(U_M) = \dim(V_M) = \rank M \le m < t$. By \Cref{prop:RS-subspace-design}, we have that
\[
  \sum_{j=1}^n \dim(V_M|_{I_j}) \ge n\dim(U_M) - \frac{\dim(U_M) (tn-b' - \dim(U_M))}{t - \dim(U_M) + 1}.
\]
For all $j \in [n]$, let $M|_{I_j}$ be the restriction of $M$ to the columns indexed by $I_j$. Since $V_M|_{I_j}$ is precisely the rowspan of $M|_{I_j}$ we have that $\dim(V_M|_{I_j}) \le \dim \operatorname{colspan}(M|_{I_j}) \le \dim(U_M \cap \F^{E_j}).$ Therefore, we have that
\begin{align*}
 - b\dim(U_M) &+ \sum_{j=1}^n \dim(U_M \cap \F^{E_j})\\ & \ge (n-b)\dim(U_M) -  \frac{\dim(U_M) (tn-b' - \dim(U_M))}{t - \dim(U_M) + 1}\\
&\ge \dim(U_M) \left[n - b - \frac{tn-b'-\dim(U_M)}{t - \dim(U_M)+1}\right]\\
&\ge n - b - \frac{tn-b'-1}{t-m+1} &\text{(since $U_M \neq 0$)}\\
&=\frac{(n-b)(t-m+1) - tn + tb + d + 1}{t-m+1}\\
&=\frac{d+1 - (m-1)(n-b)}{t-m+1} > 0,
\end{align*}
since $d = mn > (m-1)(n-b)$, as desired. This proves (\ref{eq:potential-goal}), completing the proof.
\end{proof}

\subsubsection{Tensoring Two Reed-Solomon Codes}

As a near-immediate corollary of \Cref{lem:replace-FRS}, we show that both $C_{\col}$ and $C_{\row}$ can be replaced by suitable RS codes. To do this, let $s = 8m^5n^4$ and $d' = 2m^3n^2$.

\begin{itemize}
\item Let $a' := sa + d'$.
\item Pick $\alpha_1, \hdots, \alpha_m \in \F$ and $\gamma \in \F$ such that the set $Y := \{\alpha_i \gamma^j \mid i \in [n], j \in \{0,1 , \hdots, s-1\}\}$ has exactly $sm$ elements.
\item Let $C^{\RS}_{\col} \subseteq \F^{sm}$ be the $[sm, sm - a']$ Reed-Solomon code with evaluation points from $Z$.
\end{itemize}

Again, we may set $\gamma = 2$ and $\alpha_1, \hdots, \alpha_m$ to be consecutive odd integers. Since $s = 2(2m^2n^2)^2n$ and $d' = (2m^2n^2)m$, by \Cref{lem:replace-FRS}, we have the following corollary.

\begin{corollary}\label{cor:transpose}
Let $C'_{\row}$ be any $[tn, tn-b']$ code. For all $E \subseteq [2m^2n^2] \times [m]$ we have that $E$ is an independent set of $\cM(C'_{\row}, m, a)$ if and only $E^{1,s}$ is an independent set of $\cM(C_{\row}', C^{\RS}_{\col}).$
\end{corollary}

By the theory of \Cref{subsec:potential-matroid}, we have our key structural result.

\begin{lemma}\label{lem:tensor-two}
Assuming \Cref{conj:birigid}, $E \subseteq [m] \times [n]$ is an independent set of $\MR(m, n, a, b)$ if and only if $E^{s,t}$ is an independent set of $\cM(C_{\col}^{\RS}, C^{\RS}_{\row})$.
\end{lemma} 

\begin{proof}
Pick (generic) $[m,m-a]$ code $C_\col$ and $[n,n-b]$ code $C_{\row}$ such that $\cM(C_{\col}, C_{\row}) = \MR(m,n,a,b)$. By \Cref{lemma:potential-monotone}, any independent set $E$ of $\MR(m, n, a, b)$ is an independent set of $\cM(C_{\col}, n, b)$. By \Cref{lem:replace-FRS}, we have that $E^{1,t}$ is an independent set of $\cM(C_{\col}, C_{\row}^{\RS})$. Let $F \subseteq [n] \times [m]$ be the ``transpose'' of $E$. Then, $F^{t,1}$ is an independent set of $\cM(C_{\row}^{\RS}, C_{\col})$. Thus, by \Cref{lemma:potential-monotone}, we have that $F^{t,1}$ is an independent set of $\cM(C_{\row}^{\RS}, m, a)$. By \Cref{cor:transpose} with $C'_{\row} = C^{\RS}_{\row}$, we have that $F^{t,s}$ is an independent set of $\cM(C_{\row}^{\RS}, C^{\RS}_{\col})$, so $E^{s,t}$ is an independent set of $\cM(C_{\row}^{\RS}, C^{\RS}_{\col})$.

For the other direction, assume $E$ is dependent in $\cM(C_{\col}, C_{\row})$. By \Cref{cor:potential-MR}, assuming \Cref{conj:birigid}, we have that $E$ is dependent in $\cM(C_{\col}, n, b)$.  Pick a generic $[tn, tn-b']$ code $C'_{\row}$ such that we have $\cM(C_{\col}, C'_{\row}) = \MR(m, tn, a, b')$.  By \Cref{prop:RS-dependent}, we have that $E^{1,t}$ is a dependent set of $\cM(C_{\col}, tn, b') \succeq \MR(m, tn, a, b') = \cM(C_{\col}, C'_{\row})$. Thus, $F^{t,1}$ is a dependent set of $\cM(C'_{\row}, C_{\col}) = \MR(tn, m, b', a) = \cM(C'_{\row}, m, a)$, where the latter equality holds by \Cref{cor:potential-MR} assuming \Cref{conj:birigid}. By \Cref{prop:RS-dependent}, we have that $F^{t,s}$ is a dependent set of $\cM(C'_{\row}, C_{\col}^{\RS})$. Thus, $E^{s,t}$ is a dependent set of $\cM(C^{\RS}_{\col}, C_{\row})$ and thus is a dependent set of $\cM(C^{\RS}_{\col}, C^{\RS}_{\row})$.
\end{proof}

\subsection{Proof of \Cref{thm:MR-algo}}\label{subsec:proof-thm}

We next present Algorithm~\ref{algo:MR-rank} which attempts to compute the rank function $\MR(m,n,a,b)$. In the algorithm, we assume that $\F = \Q$. The correctness of this algorithm is conditional on \Cref{conj:birigid}.

\begin{algorithm}\label{algo:MR-rank}
\caption{(Conditional) Independence Oracle}
\SetAlgoLined
\KwIn{Parameters $m,n,a,b \in \mathbb N$. Pattern $E \subseteq [m] \times [n]$.}
\KwOut{Whether $E$ is independent in $\MR(m, n, a, b)$}
{
$s := 8m^5n^4$\;
$t := 2m^2n$\;
$a' := sa + 2m^3n^2$\;
$b' := sb + mn$\;
$\alpha_1, \hdots, \alpha_m := 1, 3, \hdots, 2m-1$\;
$\beta_1, \hdots, \beta_n := 1, 3, \hdots, 2n-1$\;
$\gamma := 2$\;
$C^{\RS}_{\col} := \RS^{\Q}_{sm, sm-a'}(\alpha_1, \gamma\alpha_1, \hdots, \gamma^{s-1}\alpha_1, \hdots, \alpha_n, \gamma\alpha_n, \hdots, \gamma^{s-1}\alpha_m)$ \;
$C^{\RS}_{\row} := \RS^{\Q}_{tn, tn-b'}(\beta_1, \gamma\beta_1, \hdots, \gamma^{t-1}\beta_1, \hdots, \beta_n, \gamma\beta_n, \hdots, \gamma^{t-1}\beta_n)$ \;
\If{$E^{s,t}$\text{ is independent in }$C^{\RS}_{\col} \otimes C^{\RS}_{\row}$}{\Return{\textsf{INDEPENDENT}\;}}
\Return{\textsf{DEPENDENT}}\;
}
\end{algorithm}

To complete the proof of \Cref{thm:MR-algo}, we need to prove that Algorithm~\ref{algo:MR-rank} is correct (assuming \Cref{conj:birigid}) and that Algorithm~\ref{algo:MR-rank} it indeed runs in deterministic polynomial time.

\paragraph{Correctness.} The correctness follows immediately by \Cref{lem:tensor-two}.

\paragraph{Runtime Analysis.} We note that this algorithm runs in polynomial time unconditionally, as checking whether $E^{s,t}$ is independent in $C^{\RS}_{\col} \otimes C^{\RS}_{\row}$ is equivalent to checking that the dimension of $(C^{\RS}_{\col} \otimes C^{\RS}_{\row})|_{\overline{E^{s,t}}}$ is $(sm-a')(tn-b')$. This is equivalent to computing the rank of a matrix with $\poly(m,n)$ entries each of which is an integer representable in $\poly(m,n)$ bits. The rank of this matrix can be computed in polynomial time as Gaussian elimination can be done in strongly polynomial time over $\mathbb Q$~\cite{Sch86}.

\section{Improvements and Limitations for Subspace Designs}\label{sec:subspace-design}

We now return to $\F$ representing an arbitrary field (possibly finite). Recall that a collection of vector spaces $H_1, \hdots, H_n \subseteq \F^k$ each of codimension $s$ form an $(\ell,A)$-strong subspace design for some $\ell \le s$ if for all $\ell$-dimensional $U \subseteq \F^k$, we have that
\[
    \sum_{i=1}^n \dim(H_i \cap U) \le A.
\]
Likewise, we say that these spaces form an $(\ell, A)$-weak subspace design if for all $\ell$-dimensional $U \subseteq \F^k$, we have that
\[
    \sum_{i=1}^n \one[H_i \cap U \neq 0] \le A.
\]
As shown by this paper, and many prior works~(e.g., \cite{guruswami2016explicit,guruswami2018subspace,tamo24,cz24}), strong and weak subspace designs have proved important in near-optimal explicit constructions for various coding theory problems. Thus, we seek to better understand what is the optimal choice of $A$ for various choices of $\F,n,k,\ell$. Toward this problem, we present two results, the first improves on the well-known bound of Guruswmai and Kopparty~\cite{guruswami2016explicit}.

\begin{theorem}\label{thm:gk16-improved}
Assume that $|\F| > ns$, then there exists an explicit choice of $H_1, \hdots, H_n \subseteq \F^k$ of codimension $s$ which forms an $(\ell, \lfloor \frac{\ell(k-\ell)}{s-\ell+1}\rfloor)$ strong subspace design.
\end{theorem}

This improves on Guruswami and Kopparty's bound of $\lfloor \frac{\ell(k-1)}{s-\ell+1}\rfloor$ by using a near-identical construction and a slightly more refined analysis.

Our second result, more surprisingly, shows that \cref{thm:gk16-improved} is tight when $\F$ is an algebraically closed field and $n$ is sufficiently large.

\begin{theorem}\label{thm:gk16-optimal}
Let $\F$ be an algebraically closed field (of any characteristic), then for any $H_1, \hdots, H_n \subseteq \F^k$ of codimension $s$ (possibly with repetition), if $n \le \frac{\ell(k-\ell)}{s-\ell+1}$, then there exists an $\ell$-dimensional subspace $U \subseteq \F^k$ which has nontrivial intersection with each $H_i$.  In other words, if $n \ge \lfloor \frac{\ell(k-\ell)}{s-\ell+1}\rfloor $ then a $(\ell, \lfloor \frac{\ell(k-\ell)}{s-\ell+1}\rfloor - 1)$ weak subspace design does not exist.
\end{theorem}

The main technique toward proving \cref{thm:gk16-optimal} is to use Schubert calculus to analyze how the conditions $U \cap H_i$ affect the structure of the projective variety induced by all $\ell$-dimensional subspaces of $\F^k$ (a.k.a., the Grassmannian). We prove \cref{thm:gk16-improved} and \cref{thm:gk16-optimal} in \cref{sec:gk16-improved} and \cref{sec:gk16-optimal}, respectively. However, Theorem~\ref{thm:gk16-optimal} does not apply over non-algebraically closed fields. See \cref{sec:counterexample} for a counterexample.

\subsection{An Improved Analysis of Guruswami--Kopparty}\label{sec:gk16-improved}

Toward proving Theorem~\ref{thm:gk16-improved}, we begin by discussing the subspace design construction of Guruswami--Kopparty~\cite{guruswami2016explicit} and then explaining how to improve upon their analysis. First, observe that $\F^k$ is in bijection with the vector space of polynomials $f \in \F[x]_{<k}$ of degree less than $k$. Pick a nonzero element $\gamma \in \F$ of order at least $k$, as well as nonzero $\alpha_1, \hdots, \alpha_n$, we call this choice of $\gamma, \alpha_1, \hdots, \alpha_n$ \emph{appropriate} if all of
\[
    \{\gamma^{j-1} \alpha_i : (i,j) \in [n] \times [s]\}
\]
are distinct. Given an appropriate choice, Guruswami--Kopparty construction $H_1, \hdots, H_n \subseteq \F^k$ as follows. For each $i \in [n]$, let $H_i$ be the set of all polynomials $f \in \F[x]_{<k}$ with $\alpha_1, \gamma\alpha_1, \hdots, \gamma^{s-1}\alpha_1$ as roots. By the appropriate condition, each $H_i$ has codimension $s$.

We now seek to show that $(H_1, \hdots, H_n)$ form a strong subspace design with the parameters of \cref{thm:gk16-improved}. Fix an $\ell$-dimensional $U \subseteq \F[x]_{< k}$ as well as a polynomial basis $f_1, \hdots, f_\ell$ of $U$. Recall from \Cref{def:wronskian}, the \emph{Wronskian} of this basis is as follows
\[
W_\gamma\left(f_1, \ldots, f_\ell\right)(X) \stackrel{\text { def }}{=}\left(\begin{array}{ccc}
f_1(X) & \ldots & f_\ell(X) \\
f_1(\gamma X) & \cdots & f_\ell(\gamma X) \\
\vdots & \ddots & \vdots \\
f_1\left(\gamma^{\ell-1} X\right) & \cdots & f_\ell\left(\gamma^{\ell-1} X\right)
\end{array}\right),
\]
where $X$ is a symbolic variable (i.e., the entries of the matrix lie in $\F[X]$). Note our application is a bit more general than \Cref{def:wronskian}, as we no longer assume $\gamma$ is a multiplicative generator. The following lemma is key.
\begin{lemma}\label{lem:Wronskian}
Let $p(X) = \det W_{\gamma}(f_1, \hdots, f_\ell)(X)$. Then, $p$ has the following properties.
\begin{itemize}
\item[(a)] If the multiplicative order of $\gamma$ is at least $k$, then $p(X) \neq 0$.
\item[(b)] The degree of $p(X)$ is at most $\ell k - \binom{\ell+1}{2}$.
\item[(c)] $X^{\binom{\ell}{2}}$ is a factor of $p(X)$.
\item[(d)] If $\dim(H_i \cap U) = d_i$, then $(X-\alpha_1)^{d_i}, \hdots, (X-\gamma^{s-\ell}\alpha_1)^{d_i}$ are factors $p(X)$.
\end{itemize}
\end{lemma}

Note that (a), (b), (c) improve upon Guruswami--Kopparty in various ways, while (d) is identical. more precisely, for (a), Guruswami--Kopparty assume that $\gamma$ is a multiplicative generator of $\F$, although we relax that to just needing a lower bound on the order of $\gamma$. For (b), we improve over their bound of $\ell(k-1)$, and no observation like (c) appears in \cite{guruswami2016explicit}. Assuming \cref{lem:Wronskian}, \cref{thm:gk16-improved} follows easily.

\begin{proof}[Proof of \cref{thm:gk16-improved}]
Construct $H_1, \hdots, H_n \subseteq \F^k$ as indicated above. Consider any $\ell$-dimensional $W \subseteq \F^k$ as well as a corresponding nonzero polynomial $p(X)$ of \cref{lem:Wronskian}. By \cref{lem:Wronskian}(b), we know that $\deg p \le \ell k - \binom{\ell+1}{2}$. By \cref{lem:Wronskian}(c) $0$ is a root of $p(X)$ $\binom{\ell}{2}$ times. By \cref{lem:Wronskian}(d) we know that each of $\alpha_i, \hdots, \gamma^{s-\ell}\alpha_i$ is are roots of $p$ with multiplicity at least $\dim(H_i \cap U)$ each. Thus, since $\alpha_1, \hdots, \alpha_n, \gamma$ are appropriate, we have that
\[
    \ell k - \binom{\ell+1}{2} \ge \deg p \ge \binom{\ell}{2} + \sum_{i=1}^n (s-\ell+1)\dim(H_i \cap U).
\]
Thus, 
\[
\frac{\ell (k - \ell)}{s - \ell + 1} \ge \sum_{i=1}^n \dim(H_i \cap U),
\]
as desired.
\end{proof}

We now proceed to prove the various parts of \cref{lem:Wronskian}.

\subsubsection{Proof of \cref{lem:Wronskian}(a) and (b)}\label{subsec:ab}

This proof technique appears in \cite{guruswami2011LinearAlgebraica}, but we give the argument for completeness.\footnote{See also \url{https://www.cnblogs.com/Elegia/p/18738181/wronskian}} Recall that $W_{\gamma}(f_1, \hdots, f_\ell)(X)$ depends on the choice of basis $f_1, \hdots, f_\ell$ of $U$. However, any two bases $f_1, \hdots, f_\ell$ and $g_1, \hdots, g_\ell$ of $U$ are related by an invertible matrix $Q \in \F^{\ell \times \ell}$ where $g_i = \sum_{j = 1}^{\ell} Q_{i,j} f_j$. As such, it is not hard to see then that
\[
    W_{\gamma}(g_1, \hdots, g_\ell)(X) = W_{\gamma}(f_1, \hdots, f_\ell)(X) Q^{\top},
\]
In particular, 
\[
\det W_{\gamma}(g_1, \hdots, g_\ell)(X) = \det W_{\gamma}(f_1, \hdots, f_\ell)(X) \det Q,
\]
where $\det Q$ is a nonzero element of $\F$. Thus, for any of (a)-(d) in \cref{lem:Wronskian}, we may freely pick a basis of $U$ which is most convenient for that part. For parts (a) and (b), we consider a basis $f_1, \hdots, f_\ell$ of $U$ for which $\deg f_1 < \deg f_2 < \cdots < \deg f_\ell < k$. For each $i \in [\ell]$, let $d_i = \deg f_i$. Furthermore, let $c_1, \hdots, c_\ell \in \F^{\times}$ be the leading coefficients of $f_1, \hdots, f_\ell$. 

It is clear that if $p(X) = \det W_{\gamma}(f_1, \hdots, f_\ell)(X)$ is nonzero, then $\deg p \le d_1 + \cdots + d_\ell \le (k-\ell) + (k-(\ell-1)) + \cdots + (k-1) = \ell k - \binom{\ell+1}{2}$. Thus, (a) implies (b). To show (a), we prove that the coefficient $X^{d_1 + \cdots + d_\ell}$ of $p(X)$ is nonzero. To see why, note that the coefficient of $X^{d_1 + \cdots + d_\ell}$ in the expansion of $\det W_{\gamma}(f_1, \hdots, f_\ell)(X)$ is
\begin{align*}
    \sum_{\sigma \in S_\ell} (-1)^{\sigma} \prod_{i=1}^\ell c_i(\gamma^{\sigma(i)-1} X)^{d_i} = \gamma^{-\ell}\prod_{i=1}^\ell c_i X^{d_i} \cdot \left[\sum_{\sigma \in S_s} (-1)^{\sigma} \gamma^{d_1\sigma(1) + \cdots + d_\ell\sigma(\ell)}\right],
\end{align*}
where $S_\ell$ is the set of permutations of $[\ell]$ and $(-1)^{\sigma}$ is the sign of a particular permutation $\sigma \in S_\ell$. Now observe that

\[
\sum_{\sigma \in S_\ell} (-1)^{\sigma} \gamma^{d_1\sigma(1) + \cdots + d_\ell\sigma(\ell)} = \det \begin{pmatrix}1 & 1 & \cdots & 1\\
\gamma^{d_1} & \gamma^{d_2} & \cdots & \gamma^{d_\ell}\\
\vdots & \vdots & \ddots & \vdots\\
\gamma^{(\ell-1)d_1} & \gamma^{(\ell-1)d_2} & \cdots & \gamma^{(\ell-1)d_\ell}
\end{pmatrix}.
\]
This is nonzero only if $\gamma^{d_i} = \gamma^{d_j}$ for some $i$ and $j$. But, since the order of $\gamma$ is at least $k$, this cannot happen. Thus, we have proved (a) and also (b).

\subsubsection{Proof of \cref{lem:Wronskian}(c)}

As mentioned in \cref{subsec:ab}, we may pick a most convenient basis to prove this part. By computing a suitable row-echelon form of $U$, we may assume that our basis $f_1, \hdots, f_s$ of $U$ has the property that $f_i$ is divisible by $X^{i-1}$ for all $i \in [\ell]$. Thus, $p(X)$ must be divisible by $\prod_{i=1}^{\ell} X^{i-1} = X^{\binom{\ell}{2}}$, as desired.

\subsubsection{Proof of \cref{lem:Wronskian}(d)}

Assume that $\dim(H_i \cap U) = d_i$. Thus, there exists a basis $f_1, \hdots, f_\ell$ of $W$ for which $f_1, \hdots, f_{d_i} \in H_i$. Now for any $\beta \in \{\alpha_i, \hdots, \gamma^{s-\ell}\alpha_{i}\}$, observe that the first $d_i$ columns of $W_{\gamma}(f_1, \hdots, f_\ell)(\beta)$ are identically zero. That is, each of the first $d_i$ columns of $W_{\gamma}(f_1, \hdots, f_\ell)(X)$ are divisible by $X-\beta$. Thus, $p(X)$ must be divisible by $(X-\beta)^{d_i}$ for all $\beta \in \{\alpha_i, \hdots, \gamma^{s-\ell}\alpha_{i}\}$, as desired.

\begin{remark}
If $\dim(H_i \cap U) \ge 2$, then we can identify additional roots of $p(X)$ (e.g., $\gamma^{-1}\alpha_i$). This allow for a tighter analysis in various applications, although we defer such arguments to future work.
\end{remark}

\subsubsection{Proof of \Cref{prop:RS-subspace-design}}\label{subsec:deferred-proof}

We now prove \Cref{prop:RS-subspace-design}, which we state in more generality as follows.

\begin{proposition}\label{prop:RS-subspace-design-better}
Let $C^{\RS} \subseteq \F^{sn}$ be an $[sn,k]$ Reed-Solomon code with evaluation points \[(\alpha_1, \gamma \alpha_1, \hdots, \gamma^{s-1}\alpha_1, \hdots, \alpha_n, \gamma \alpha_n, \hdots, \gamma^{s-1}\alpha_n).\]  
For $i \in [n]$, let $I_i := \{s(i-1)+1, s(i-1)+2, \hdots, si\}$. For any subcode $V \subseteq C^{\RS}$ for which $\dim(V) \le s$, we have that
\begin{align}
  \sum_{i=1}^n \dim(V|_{I_i}) \ge n\dim(V) - \frac{\dim(V) (k - \dim(V))}{s - \dim(V) + 1}.\label{eq:puncture-bound}
\end{align}
\end{proposition}

\begin{proof}
Let $\psi : \F_{<k}[x] \to C^{\RS}$ be the canonical encoding map of $\C^{RS}$. That is, 
\[
\psi(f) = (f(\alpha_1), f(\gamma \alpha_1), \hdots, f(\gamma^{s-1}\alpha_1), \hdots, f(\alpha_n), f(\gamma \alpha_n), \hdots, f(\gamma^{s-1}\alpha_n)).
\]
Since $V \subseteq C^{\RS}$ is a subcode, there exists a unique $U \subseteq \F_{<k}[x]$ with $\dim(U) = \dim(V)$ for which $\psi(U) = V$.

Let $H_1, \hdots, H_n \subseteq \F_{<k}[x]$ be spaces of codimension $s$ such that $H_i$ corresponds to polynomials of degree less than $k$ with $\{\alpha_i, \gamma \alpha_i, \hdots, \gamma^{s-1} \alpha_i\}$ as roots. By Theorem~\ref{thm:gk16-improved}, we have that
\[
\sum_{i=1}^n \dim(H_i \cap U) \le \frac{\dim(U) (k - \dim(U))}{s - \dim(U) + 1} = \frac{\dim(V) (k - \dim(V))}{s - \dim(V) + 1}. 
\]
Thus, to prove (\ref{eq:puncture-bound}), it suffices to prove that $\dim(V|_{I_i}) = \dim(U) - \dim(U \cap H_i)$ for all $j \in [n]$. Given $u \in U$, we have that $\psi(u)|_{I_i} = 0^{s}$ if and only if $u \in H_i$.  In other words, $\dim V|_{I_i} = \dim \psi(U)|_{I_i} = \dim(U) - \dim(H_i \cap U)$, as desired. 
\end{proof}

\subsection{Lower Bounds over Algebraically Closed Fields}\label{sec:gk16-optimal}

We now turn toward proving limitations on subspace designs. To prove \cref{thm:gk16-optimal}, we first need some machinery from algebraic geometry. 

\subsubsection{Grassmaniann and Pl\"ucker Coordinates}

Note that a subspace design is a list of algebraic conditions for all $\ell$-dimensional $W \subseteq \F^k$. The set of all such subspaces is known as the \emph{Grassmannian} $\Gr(\ell, \F^k)$ (see \cite{lakshmibai2015Grassmannian} for a comprehensive reference). To study $\Gr(\ell, \F^k)$ algebraically, we use what are known as \emph{Pl\"ucker coordinates}. More precisely, we can identify any $\ell$-dimensional $U \subseteq \F^k$ with a matrix $M_{U} \in \F^{\ell \times k}$ whose $\ell$ rows form a basis of $U$. For each $S \in \binom{[k]}{\ell}$, we define the $S$th Pl{\"u}cker coordainte to be $p_S(U) := \det (M_{U}|_{S})$, where $|_S$ is the restriction to the columns indexed by $S$. Now, observe that different bases $M_U$ of $U$ may produce different Pl\"ucker coordinates. However, one can show that any two sets of Pl{\"u}cker coordiantes of $W$ have the same up to a global nonzero scalar (i.e., $p_S(U) = \lambda q_S(U)$ for some $\lambda \in \F^{\times}$ for all $S \in \binom{[k]}{\ell}$). In other words, we think of $\Gr(\ell, \F^k)$ as a \emph{projective} variety. A crucial fact we shall use is that the dimension of $\Gr(\ell, \F^k)$ is $\ell(k-\ell)$~\cite{harris1992Algebraic,lakshmibai2015Grassmannian}.

\subsubsection{Schubert Variety}

Recall that a weak subspace design consists of a number of constraints of the form $\dim(H \cap U) \ge 1$, where $H$ is a fixed subspace and $U$ is selected from the Grassmannian $\Gr(\ell,\F^k)$. The set $V_H \subseteq \Gr(\ell, \F^k)$ of $U$ which satisfy this condition is a special case of a \emph{Schubert (sub)variety}. 

\begin{definition}[Schubert Variety, (e.g., \cite{lakshmibai2015Grassmannian})]
Given a vector space $\F^k = \langle e_1, \hdots, e_k\rangle$ and an integer $\ell \in \{0,1, \hdots, k\}$, consider a sequence $\vec{a} = (a_1, \hdots, a_\ell)$ be a sequence of integers $1 \le a_1 < a_2 < \cdots < a_\ell \le k$. We then define
\[
    \Sch_{\vec{a}}(\ell, e_1,\hdots, e_k) := \{U \in \Gr(\ell, \F^k) : \forall i \in [\ell], \dim(W \cap \langle e_j : j \in [a_i] \rangle) \ge i\}.
\]
\end{definition}
In particular, if $\vec{a} = (k-\ell+1, k-\ell+2, \hdots, k)$, then $\Sch_{\vec{a}}(\ell, e_1, \hdots, e_k) = \Gr(\ell, \F^k)$ because every $U \in \Gr(\ell, \F^k)$ satisfies 
\begin{align}
\dim(U \cap \langle e_j : j \in [k-\ell+i]\rangle) &\ge  \dim(U) + \dim(\langle e_j : j \in [k-\ell+i]\rangle) - k\nonumber\\
&= \ell + k - \ell + i - k \ge i. \label{eq:W-pigeon}
\end{align}

Besides the fact that each $\Sch_{\vec{a}}$ is a projective variety, we also need the following formula for the (Krull) dimension (see \cite{hartshorne1977Algebraic} for a precise definition) of each variety.

\begin{theorem}[e.g., Theorem 5.3.7~\cite{lakshmibai2015Grassmannian}]\label{thm:Schubert-dim}
For all $\vec{a} = (a_1, \hdots, a_\ell)$ with $1 \le a_1 < a_2 < \cdots < a_\ell \le k$, we have that
\[
    \dim \Sch_{\vec{a}}(\ell, e_1,\hdots, e_k) = \sum_{i=1}^\ell (a_i - i).
\]
\end{theorem}
In particular, the choice $\vec{a} = (k-\ell+1, k-\ell+2, \hdots, k)$ implies that $\dim \Gr(\ell, \F^k) = \ell(k-\ell)$ (Corollary 5.3.8~\cite{lakshmibai2015Grassmannian}). As an immediate corollary of Theorem~\ref{thm:Schubert-dim}, we can compute the dimension of each $V_H := \{U \in \Gr(\ell, \F^k) : V_H \cap U \neq 0\}$ we need for analyzing subspace varieties.

\begin{corollary}\label{cor:design-dim}
For any $H \subseteq \F^k$ of dimension $k-s$ with $s \ge \ell$, we have that $V_H = \{U \in \Gr(\ell, \F^k) : V_H \cap U \neq 0\}$ is projective variety of dimension $\ell(k-\ell) - (s-\ell+1)$.
\end{corollary}
\begin{proof}
Let $b_1, \hdots, b_k$ be a basis of $\F^k$ such that $b_1, \hdots, b_{k-s}$ is a basis of $H$. Let $\vec{a} = (k-s, k-\ell+2, \cdots, k)$. We claim that $\Sch_{\vec{a}}(\ell, b_1, \hdots, b_k) = V_H$. To see why, note that for $i \ge 2$, the condition that $\dim(U \cap \langle b_j : j \in [a_i] \rangle) \ge i$ in $\Sch_{\vec{a}}$ is immediate from (\ref{eq:W-pigeon}). Furthermore, by the choice of basis, the condition $\dim(U \cap \langle b_j : j \in [k-s]\rangle) \ge 1$ is exactly $\dim(U \cap H) \ge 1$, i.e., $U \cap H \neq 0$. Thus, $V_H$ is a projective subvariety of $\Gr(\ell, \F^k)$.

To finish, we apply Theorem~\ref{thm:Schubert-dim} to get that
\[
    \dim(V_H) = a_1 -1 + \sum_{i=2}^{\ell} (a_i - i) = (k-s-1) + (\ell-1)(k-\ell) = \ell(k-\ell) - (s-\ell+1),
\]
as desired.
\end{proof}

\subsubsection{Proof of \cref{thm:gk16-optimal}}

To prove Theorem~\ref{thm:gk16-optimal}, we use a standard fact of algebraic geometry that dimensions of intersections of projective varieties behave analogously to that of vector spaces.

\begin{proposition}[Lemma 43.13.14~\cite{stacks-project} and Theorem I.7.2~\cite{hartshorne1977Algebraic}, adapted]\label{prop:dim-intersect}
Let $X, Y, Z$ be projective varieties over an algebraically closed field with $X, Y \subseteq Z$. If $\dim(X) + \dim(Y) \ge \dim(Z)$, then $X \cap Y$ is nonempty and has dimension\footnote{Technically, the result applies to the each irreducible component of the respective varieties, but here we define the dimension of a variety to be the maximum dimension of its irreducible components.} at least $\dim(X) + \dim(Y) - \dim(Z)$.
\end{proposition}

\begin{theorem}[Theorem~\ref{thm:gk16-optimal} restated]
Let $\F$ be an algebraically closed field (of any characteristic), then for any $H_1, \hdots, H_n \subseteq \F^k$ of codimension $s$ (possibly with repetition), if $n \le \frac{\ell(k-\ell)}{s-\ell+1}$, then there exists an $\ell$-dimensional subspace $U \subseteq \F^k$ which has nontrivial intersection with each $H_i$.  In other words, if $n \ge \lfloor \frac{\ell(k-\ell)}{s-\ell+1}\rfloor $ then a $(\ell, \lfloor \frac{\ell(k-\ell)}{s-\ell+1}\rfloor - 1)$ weak subspace design does not exist.
\end{theorem}

\begin{proof}[Proof of Theorem~\ref{thm:gk16-optimal}]
Let $V_{H_i}$ be the subvariety of $\Gr(\ell, \F^k)$ corresponding to the $U$ for which $H_i \cap U \neq 0$. It suffices to prove that if $n \le \frac{\ell(k-\ell)}{s-\ell+1}$ then $\bigcap_{i=1}^n V_{H_i} \neq \emptyset$. By Corollary~\ref{cor:design-dim}, we have that $\dim(V_{H_i}) = \ell(k-\ell)-(s-\ell+1)$ for all $i \in [n]$. We claim by induction for all $m \in [n]$ that $\bigcap_{i=1}^m V_{H_i}$ is nonempty with each irreducible component having dimension at least $\ell(k-\ell)-m(s-\ell+1)$. The base case of $m=1$ is precisely Corollary~\ref{cor:design-dim}. For $m \ge 2$, assume the inductive hypothesis for $m-1$ and apply Proposition~\ref{prop:dim-intersect} with $X = \bigcap_{i=1}^{m-1} V_{H_i}$, $Y = V_{H_m}$ and $Z = \Gr(\ell, \F^k)$. Then, observe that
\[
    \dim(X) + \dim(Y) - \dim(Z) \ge \ell(k-\ell)-(m-1)(s-\ell+1) + \ell(k-\ell)-(s-\ell+1) - \ell(k-\ell) = \ell(k-\ell)-m(s-\ell+1).
\]
Therefore, since $m \le n \le \frac{\ell(k-\ell)}{s-\ell+1}$, we have that $\dim(X) + \dim(Y) - \dim(Z) \ge 0$, so $\bigcap_{i=1}^m V_{H_i}$ is nonempty with dimension at least $\ell(k-\ell)-m(s-\ell+1).$ In particular, when $m = n$, we have that $\bigcap_{i=1}^n V_{H_i}$ is nonempty, so $(H_1, \hdots, H_n)$ cannot be an $(\ell, n-1)$ weak subspace design. 
\end{proof}

\subsection{Theorem~\ref{thm:gk16-optimal} does not apply to $\F_3$}\label{sec:counterexample}

We now show that \cref{thm:gk16-optimal} does not apply when $\F$ is not algebraically closed. Let $e_1, \hdots, e_4$ be a basis of $\F_3^4$ and let
\begin{align*}
H_1 &= \langle e_1, e_2\rangle\\
H_2 &= \langle e_3, e_4\rangle\\
H_3 &= \langle e_1+e_3, e_2+e_4\rangle\\
H_4 &= \langle e_1+e_4, e_2-e_3\rangle.
\end{align*}

\begin{proposition}
$H_1, \hdots, H_4$ form a $(2,3)$-weak subspace design.
\end{proposition}
Thus, Theorem~\ref{thm:gk16-optimal} cannot be extended to $\F_3$

\begin{proof}
We seek to show that there is no 2-dimensional $U \subseteq \F_3^4$ which intersects every $H_i$. Assume for sake of contradiction that there exists a $U$ for which $H_i \cap W \neq 0$. Since $H_1$ and $H_2$ are disjoint and $U$ is $2$-dimensional, there exist $\lambda_1, \hdots, \lambda_4 \in \F_3$ with $(\lambda_1, \lambda_2), (\lambda_3, \lambda_4) \neq (0,0)$ such that
\[
    U = \langle \lambda_1 e_1 + \lambda_2 e_2, \lambda_3 e_3 + \lambda_4 e_4\rangle.
\]
Now, in order for $U \cap H_3 \neq 0$, we must have that $(\lambda_1, \lambda_2)$ is a scalar multiple of $(\lambda_3, \lambda_4)$. Since $U$ is invariant to the scaling of its basis, we may assume without loss of generality that $\lambda_3 = \lambda_1$ and $\lambda_4 = \lambda_2$. Now, since $U \cap H_4 \neq 0$, there exists $a, b, c, d \in \F_3$ with $(a,b) \neq (0,0)$ and $(c,d) \neq (0,0)$ such that
\[
    a(\lambda_1 e_1 + \lambda_2 e_2) + b(\lambda_1 e_3 + \lambda_2 e_4) = c(e_1+e_4) + d(e_2 - e_3).
\]
As such, we get the following system of equations
\begin{align*}
    a\lambda_1 &= c\\
    a\lambda_2 &= d\\
    b\lambda_1 &= -d\\
    b\lambda_2 &= c.
\end{align*}
Thus, $a\lambda_1 = b\lambda_2$ and $a\lambda_2 = -b\lambda_1$. Since $(\lambda_1, \lambda_2) \neq (0,0)$. We break into cases.

First, if $\lambda_1 \neq 0$, we have that $a = b\lambda_2/\lambda_1$, so we get that $b\lambda_2^2 = -b\lambda_1^2$. If $b = 0$, then $a = 0$ which contradicts that $(a,b) \neq (0,0)$. Thus, $b \neq 0$ so $\lambda_1^2 + \lambda_2^2 = 0$, but the only solution to this equation over $\F_3$ is $(\lambda_1, \lambda_2) = (0,0)$, a contradiction.

Likewise, if $\lambda_2 \neq 0$, we get $a = -b\lambda_1/\lambda_2$ so $-b\lambda_1^2 = b\lambda_2^2$. We then proceed by identical logic to the previous case. Thus, $(H_1, \hdots, H_4)$ is a $(2,3)$ weak subspace design.
\end{proof}

\begin{remark}
If we look at a quadratic extension of $\F_3$, then the equation $\lambda_1^2 + \lambda_2^2=0$ has nonzero solutions, from which we can then construct a subspace $U$ intersecting all four spaces.
\end{remark}

\section*{Acknowledgments} We thank Sivakanth Gopi for numerous discussions on the relationship between \cite{lms25} and MR tensor codes, which greatly inspired \Cref{sec:matroid}. We thank Visu Makam for valuable discussions on MR tensor codes. We thank Matt Larson for helpful clarifications on the relationship between \Cref{thm:MR-algo} and the matroid/rigidity literature. We thank Mahdi Cheraghchi, Venkatesan Guruswami and Madhur Tulsiani for many helpful discussions and much encouragement.

Joshua Brakensiek  was partially supported by (Venkatesan Guruswami's) Simons Investigator award and National Science Foundation grants No. CCF-2211972 and No. DMS-2503280. Yeyuan Chen was partially supported by the National Science Foundation grant No. CCF-2236931. Zihan Zhang was partially supported by the National Science Foundation grant No. CCF-2440926.

\bibliography{main}

\end{document}